\documentclass{amsart}

\usepackage[utf8]{inputenc}
\usepackage{amssymb}
\usepackage{amsmath}
\usepackage{mathtools}
\usepackage{tikz}
\usepackage{paralist} 
\usepackage{eqnarray}
\usepackage{phdalgo}
\usepackage{hypergraph}
\usepackage{url}
\usepackage[numbers]{natbib}

\usetikzlibrary{shapes.geometric} 

\DeclareSymbolFont{stmry}{U}{stmry}{m}{n}
\DeclareMathSymbol\mapsfromchar\mathrel{stmry}{"5B}
\def\longmapsfrom{\longleftarrow\mapsfromchar}

\newcommand{\HH}{\mathcal{H}}
\newcommand{\HHbar}{\overline{\mathcal{H}}}
\newcommand{\VV}{\mathcal{V}}
\newcommand{\FF}{\mathcal{F}}
\newcommand{\RR}{\mathcal{R}}

\newcommand{\supersetvtx}{\mathit{superset}}
\newcommand{\reachable}{\rightsquigarrow}
\newcommand{\reach}{\reachable}
\newcommand{\reacheq}{\equiv}

\newcommand{\emptystack}{[\,]}

\newcommand{\incomp}{\mathit{Finished}}
\newcommand{\troot}{\mathit{low}}
\newcommand{\tindex}{\mathit{index}}
\newcommand{\ismax}{\mathit{is\_term}}

\newcommand{\cur}{\mathit{cur}}
\newcommand{\old}{\mathit{old}}
\newcommand{\graph}{\mathsf{graph}}
\newcommand{\single}{\mathit{no\_merge}}
\newcommand{\scc}{\textsc{Scc}}
\newcommand{\collected}{\mathit{collected}}
\newcommand{\new}{\mathit{new}}
\newcommand{\Find}{\Call{Find}}
\newcommand{\tplib}{{TPLib}}
\newcommand{\vtrue}{\mathbf{t}}
\newcommand{\vfalse}{\mathbf{f}}

\newcommand{\etal}{\textit{et al.}}
\newcommand{\etc}{\textit{etc}}
\newcommand{\ie}{\textit{i.e.}}

\newtheorem{proposition}{Proposition}
\newtheorem{lemma}[proposition]{Lemma}
\newtheorem{theorem}[proposition]{Theorem}
\newtheorem{corollary}[proposition]{Corollary}
\newtheorem{invariant}{Invariant}
\newtheorem{open_problem}{Question}
\theoremstyle{remark}
\newtheorem{remark}{Remark}
\newtheorem{example}[remark]{Example}

\DeclarePairedDelimiter\card{\lvert}{\rvert}
\DeclareMathOperator{\size}{\mathsf{size}}

\algrenewcommand\alglinenumber[1]{\scriptsize #1:}
\renewcommand{\lineref}[1]{{\scriptsize\ref{#1}}}

\title[Strongly connected components in directed hypergraphs]{On the complexity of strongly connected components in directed hypergraphs}
\author{Xavier Allamigeon}
\address{INRIA Saclay -- Ile-de-France and CMAP, Ecole Polytechnique, France}
\email{firstname.lastname@inria.fr}

\begin{document}
\begin{abstract}
We study the complexity of some algorithmic problems on directed hypergraphs and their strongly connected components (\scc{}s). The main contribution is an almost linear time algorithm computing the terminal strongly connected components (\ie\ \scc{}s which do not reach any components but themselves). \emph{Almost linear} here means that the complexity of the algorithm is linear in the size of the hypergraph up to a factor $\alpha(n)$, where $\alpha$ is the inverse of Ackermann function, and $n$ is the number of vertices. 
Our motivation to study this problem arises from a recent application of directed hypergraphs to computational tropical geometry.

We also discuss the problem of computing all \scc{}s. We establish a superlinear lower bound on the size of the transitive reduction of the reachability relation in directed hypergraphs, showing that it is combinatorially more complex than in directed graphs. 
Besides, we prove a linear time reduction from the well-studied problem of finding all minimal sets among a given family to the problem of computing the \scc{}s. Only subquadratic time algorithms are known for the former problem. These results strongly suggest that the problem of computing the \scc{}s is harder in directed hypergraphs than in directed graphs.
\end{abstract}

\date{\today}

\maketitle

\section{Introduction}\label{sec:introduction}

Directed hypergraphs consist in a generalization of directed graphs, in which the tail and the head of the arcs are sets of vertices. Directed hypergraphs have a very large number of applications, since hyperarcs naturally provide a representation of implication dependencies. Among others, they are used to solve several problems related to satisfiability in propositional logic, in particular on Horn formulas, see for instance~\cite{Ausiello91,Ausiello97,Gallo95,Gallo98,Pretolani03}. They also appear in problems relative to network routing~\cite{Pretolani00}, functional dependencies in database theory~\cite{AusielloJACM83}, 
model checking~\cite{Liu98}, chemical reaction networks~\cite{Ozturan08}, transportation networks~\cite{Nguyen89,Nguyen98}, and more recently, tropical convex geometry~\cite{AllamigeonGaubertGoubaultDCG2013,AllamigeonGaubertGoubaultSTACS10}.

Many algorithmic aspects of directed hypergraphs have been studied, in particular optimization related ones, such as determining shortest paths~\cite{Nguyen89,NielsenORL06}, maximum flows, minimum cardinality cuts, or minimum weighted hyperpaths (we refer to the surveys of Ausiello~\etal~\cite{Ausiello01} and of Gallo~\etal~\cite{GalloDAM93} for a comprehensive list of contributions). Naturally, some problems raised by the reachability relation in directed hypergraphs have also been studied. For instance, determining the set of the vertices reachable from a given vertex is known to be solvable in linear time in the size of the directed hypergraph (see for instance~\cite{GalloDAM93}).\footnote{In the sequel, the underlying model of computation is the Random Access Machine.} 

In directed graphs, many other problems can be solved in linear time, such as testing acyclicity or strong connectivity, computing the strongly connected components (\scc{}s), determining a topological sorting over them, \etc{}. Surprisingly, the analogues of these elementary problems in directed hypergraphs have not received any particular attention (as far as we know). Unfortunately, none of the direct graph algorithms can be straightforwardly extended to directed hypergraphs. The main reason is that the reachability relation of hypergraphs does not have the same structure: for instance, establishing that a given vertex $u$ reaches another vertex $v$ generally involves vertices which do not reach $v$. Moreover, as shown by Ausiello~\etal\ in~\cite{AusielloISCO12}, the vertices of a hypercycle do not necessarily belong to a same strongly connected component. 

Naturally, the aforementioned problems can be solved by determining the whole graph of the reachability relation, calling a linear time reachability algorithm on every vertex of the directed hypergraph. This naive approach is obviously not optimal, in particular when the hypergraph coincides with a directed graph.

\paragraph{Contributions} We first present in Section~\ref{sec:maxscc} an algorithm able to determine the terminal strongly connected components of a directed hypergraph in time complexity $O(N \alpha(n))$, where $N$ is the size of the hypergraph, $n$ the number of vertices, and $\alpha$ is the inverse of the Ackermann function. An \scc{} is said to be \emph{terminal} when no other \scc{} is reachable from it. The time complexity is said to be \emph{almost linear} because $\alpha(n) \leq 4$ for any practical value of $n$. As a by-product, the following two properties:
\begin{inparaenum}[(i)]
\item is a directed hypergraph strongly connected? 
\item does a hypergraph admit a sink (\ie\ a vertex reachable from all vertices)? 
\end{inparaenum}
can be determined in almost linear time. 

Problems involving terminal \scc{}s have important applications in computational tropical geometry. In particular, the algorithm presented here is the cornerstone of an analog of the double description method in tropical algebra~\cite{AllamigeonGaubertGoubaultDCG2013}. We refer to Section~\ref{subsec:other_properties} for further details, where other applications to Horn formulas and nonlinear spectral theory are also discussed. 

The contributions presented in Section~\ref{sec:combinatorics} indicate that the problem of computing the complete set of \scc{}s is very likely to be harder in directed hypergraphs than in directed graphs. 

In Section~\ref{subsec:transitive_reduction}, we establish a lower bound result which shows that the size of the transitive reduction of the reachability relation may be superlinear in the size of the directed hypergraph (whereas it is linearly upper bounded in the setting of directed graphs). An important consequence is that any algorithm computing the \scc{}s in directed hypergraphs by exploring the entire reachability relation, or at least a transitive reduction, has a superlinear complexity. 

In Section~\ref{subsec:set_pb_reduction}, we prove a linear time reduction from the minimal set problem to the problem of computing the strongly connected components. Given a family $\FF$ of sets over a certain domain, the minimal set problem consists in determining all the sets of $\FF$ which are minimal for the inclusion. While it has received much attention (see Section~\ref{subsec:set_pb_reduction} and the references therein), the best known algorithms are only subquadratic time.

\paragraph{Related Work} Reachability in directed hypergraphs has been defined in different ways in the literature, depending on the context and the applications. The reachability relation which is discussed here is basically the same as in~\cite{AusielloTCS90,Ausiello91,Ausiello01}, but is referred to as \emph{$B$-reachability} in~\cite{GalloDAM93,Gallo95}. It precisely captures the logical implication dependencies in Horn propositional logic, and also the functional dependencies in the context of relational databases. 
Some variants of this reachability relation have been introduced, for instance with the additional requirement that every hyperpath has to be provided with a linear ordering over the alternating sequence of its vertices and hyperarcs~\cite{ThakurTripathiTCS09}. These variants are beyond the scope of the paper.

As mentioned above, determining the set of the reachable vertices from a given vertex has been thoroughly studied. Gallo~\etal\  provide a linear time algorithm in~\cite{GalloDAM93}. In a series of works~\cite{AusielloTCS90,Ausiello91,Ausiello97}, Ausiello~\etal\  introduce online algorithms maintaining the set of reachable vertices, or hyperpaths between vertices, under hyperarc insertions/deletions. 

Computing the transitive closure and reduction of a directed hypergraph has also been studied by Ausiello~\etal\  in~\cite{Ausiello86}. In their work, reachability relations between sets of vertices are also taken into account, in contrast with our present contribution in which we restrict to reachability relations between vertices. The notion of transitive reduction in~\cite{Ausiello86} is also different from the one discussed here (Section~\ref{subsec:transitive_reduction}). More precisely, the transitive reduction of~\cite{Ausiello86} rather corresponds to minimal hypergraphs having the same transitive closure (several minimality properties are studied, including minimal size, minimal number of hyperarcs, \etc{}). In contrast, we discuss here the transitive reduction of the reachability relation (as a binary relation over vertices) and not of the hypergraph itself.

\section{Preliminary definitions and notations}\label{sec:preliminaries}

A \emph{directed hypergraph} is a pair $(\VV,A)$, where $\VV$ is a set of vertices, and $A$ a set of hyperarcs. A \emph{hyperarc} $a$ is itself a pair $(T,H)$, where $T$ and $H$ are both non-empty subsets of $\VV$. They respectively represent the \emph{tail} and the \emph{head} of $a$, and are also denoted by $T(a)$ and $H(a)$. Note that throughout this paper, the term \emph{hypergraph(s)} will always refer to directed hypergraph(s).

The size of a directed hypergraph $\HH = (\VV,A)$ is defined as $\size(\HH) = \card{\VV} + \sum_{(T,H) \in A} (\card{T} + \card{H})$ (where $\card{S}$ denotes the cardinality of any set $S$).

Given a directed hypergraph $\HH = (\VV,A)$ and $u,v \in \VV$, the vertex $v$ is said to be \emph{reachable} from the vertex $u$ in $\HH$, which is denoted $u \reachable_\HH v$, if $u = v$, or there exists a hyperarc $(T,H)$ such that $v \in H$ and all the elements of $T$ are reachable from $u$. This also leads to a notion of hyperpaths: a \emph{hyperpath} from $u$ to $v$ in $\HH$ is a sequence of $p$ hyperarcs $(T_1,H_1),\dots,(T_p,H_p) \in A$ satisfying $T_i \subseteq \cup_{j = 0}^{i-1} H_j$ for all $i = 1, \dots, p+1$, with the conventions $H_0 = \{u\}$ and $T_{p+1} = \{v\}$. The hyperpath is said to be \emph{minimal} if none of its subsequences is a hyperpath from $u$ to $v$. 

The \emph{strongly connected components} (\scc{}s for short) of a directed hypergraph $\HH$ are the equivalence classes of the relation $\reacheq_\HH$, defined by $u \reacheq_\HH v$ if $u \reach_\HH v$ and $v \reach_\HH u$. A component $C$ is said to be \emph{terminal} if for any $u \in C$ and $v \in \VV$, $u \reach_\HH v$ implies $v \in C$.

If $f$ is a function from $\VV$ to an arbitrary set, the image of the directed hypergraph $\HH$ by $f$ is the hypergraph, denoted $f(\HH)$, consisting of the vertices $f(v)$ ($v \in \VV$) and the hyperarcs $(f(T(a)),f(H(a)))$ ($a \in A$), where $f(S) := \{f(x) \mid x \in S\}$. 

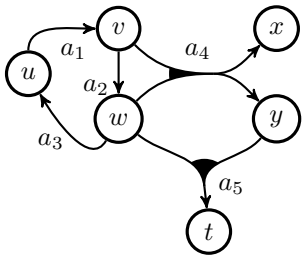
\begin{figure}[t]
\begin{center}
\begin{tikzpicture}[>=stealth',scale=0.6, vertex/.style={circle,draw=black,very thick,minimum size=2ex}, hyperedge/.style={draw=black,thick}, simpleedge/.style={draw=black,thick}]
\node [vertex] (u) at (-2,-1) {$u$};
\node [vertex] (v) at (0,0) {$v$};
\node [vertex] (w) at (0,-2)  {$w$};
\node [vertex] (x) at (3.5,0) {$x$};
\node [vertex] (y) at (3.5,-2) {$y$};
\node [vertex] (t) at (2,-4.5) {$t$};

\path[->] (u) edge[simpleedge,out=90,in=-180] (v);
\node at (-1,-0.5) {$a_1$};
\path[->] (v) edge[simpleedge,out=-90,in=90] (w);
\node at (-0.5,-1.3) {$a_2$};
\path[->] (w) edge[simpleedge,out=-120,in=-60] (u);
\node at (-1.5,-2.5) {$a_3$};
\node at (1.75,-0.5) {$a_4$};
\node at (2.5,-3.5) {$a_5$};
\hyperedge[0.5][$(hyper@tail)!0.6!(hyper@head)$]{v,w}{y,x};
\hyperedge[0.4][$(hyper@tail)!0.6!(hyper@head)$]{y,w}{t};
\end{tikzpicture}
\end{center}
\caption{A directed hypergraph}\label{fig:hypergraph}
\end{figure}

\begin{example}
Consider the directed hypergraph depicted in Figure~\ref{fig:hypergraph}. Its vertices are $u, v, w, x, y, t$, and its hyperarcs $a_1 =(\{u\}, \{v\})$, $a_2 = (\{v\}, \{w\})$, $a_3 = (\{w\}, \{u\})$, $a_4 = (\{v,w\}, \{x,y\})$, and $a_5 = (\{w,y\}, \{t\})$. A hyperarc is represented as a bundle of arcs, and is decorated with a solid disk portion when its tail contains several vertices.

Applying the recursive definition of reachability from the vertex $u$ discovers the vertices $v$, then $w$, which leads to the two vertices $x$ and $y$ through the hyperarc $a_4$, and finally $t$ through $a_5$. The vertex $t$ is reachable from $u$ through the hyperpath $a_1,a_2,a_4,a_5$ (which is minimal). As mentioned in Section~\ref{sec:introduction}, some vertices play the role of ‘‘auxiliary'' vertices when determining reachability. In our example, establishing that $t$ is reachable from $u$ requires to establish that $y$ is reachable from $u$, while $y$ does not reach $t$. Such a situation cannot occur in directed graphs.
\end{example}

Observe that all the notions presented in this section are generalizations of their analogues on directed graphs. Indeed, any digraph $G = (\VV,A)$ ($A \subseteq \VV \times \VV$) can be equivalently seen as a directed hypergraph $\HH = \bigl(\VV,\bigl\{(\{u\},\{v\}) \mid (u,v) \in A \bigr\}\bigr)$. The reachability relations on $G$ and $\HH$ coincide, and $G$ and $\HH$ both have the same size. The notations introduced here will be consequently used for directed graphs as well.

\section{Computing the terminal \scc{}s in almost linear time}\label{sec:maxscc}

\subsection{Principle of the algorithm}\label{subsec:maxscc_principle}

Given a directed hypergraph $\HH = (\VV, A)$, an hyperarc $a \in A$ is said to be \emph{simple} when $\card{T(a)} = 1$. Such hyperarcs generate a directed graph, denoted by $\graph(\HH)$, defined as the couple $(\VV, A')$ where $A' =  \{ (t,h) \mid (\{ t \},H)  \in A \text{ and } h \in H \}$.
We first point out a remarkable special case in which the terminal \scc{}s of the directed hypergraph $\HH$ and the digraph $\graph(\HH)$ are equal.
\begin{proposition}\label{prop:terminal_scc}
Let $\HH$ be a directed hypergraph. Every terminal strongly connected component of $\graph(\HH)$ reduced to a singleton is a terminal strongly connected component of $\HH$.

Besides, if all terminal strongly connected components of $\graph(\HH)$ are singletons, then $\HH$ and $\graph(\HH)$ have the same terminal strongly connected components.
\end{proposition}

\begin{proof}
Assume $\HH = (\VV,A)$. 
Let $\{ u \}$ be a terminal \scc{} of $\graph(\HH)$. Suppose that there exists $v \neq u$ such that $u \reach_\HH v$. There is necessarily a hyperarc $(T, H) \in A$ such that $T = \{ u \}$ and $H \neq \{ u \}$. Let $w \in H \setminus \{ u \}$. Then $(u,w)$ is an arc of $\graph(\HH)$. Since $\{ u \}$ is a terminal \scc{} of $\graph(\HH)$, this enforces $w = u$, which is a contradiction. Hence $\{ u \}$ is a terminal \scc{} of the hypergraph~$\HH$.

Assume that every terminal \scc{} of $\graph(\HH)$ is a singleton. Let $C$ be a terminal \scc{} of $\HH$, and $u \in C$. Consider $\{v\}$ a terminal \scc{} of $\graph(\HH)$ such that $u \reach_{\graph(\HH)} v$. Using the first part of the proof, $\{v\}$ is a terminal \scc{} of $\HH$. Besides, $\{v\}$ is reachable from $C$ in $\HH$. We conclude that $C = \{v\}$.
\end{proof}

The following proposition ensures that, in a directed hypergraph, merging two vertices of a same \scc{} does not alter the reachability relation.
\begin{proposition}\label{prop:collapse}
Let $\HH = (\VV,A)$ be a directed hypergraph, and let $x,y \in \VV$ such that $x \reacheq_\HH y$. Consider the function $f$ mapping any vertex distinct from $x$ and $y$ to itself, and both $x$ and $y$ to a same vertex $z$ (with $z \not \in \VV$). Then $u \reach_\HH v$ if, and only if, $f(u) \reach_{f(\HH)} f(v)$.
\end{proposition}

\begin{proof}
Let $\HH' = f(\HH)$. First assume that $u \reach_\HH v$, and let us show by induction that $f(u) \reach_{f(\HH)} f(v)$. The case $u = v$ is trivial. If there exists $(T, H) \in A$ such that $v \in H$ and for all $w \in T$, $u \reach_\HH w$, then $f(u) \reach_{f(\HH)} f(w)$ by induction, which proves that $f(v)$ is reachable from $f(u)$ in $f(\HH)$.

Conversely, suppose $f(u) \reach_{f(\HH)} f(v)$. If $f(u) = f(v)$, then either $u = v$, or the two vertices $u$ and $v$ belong to $\{x,y\}$. In both cases, $v$ is reachable from $u$ in $\HH$. Now suppose that there exists a hyperarc $(f(T), f(H))$ in $f(\HH)$ such that $f(v) \in f(H)$, and for all $w \in T$, $f(u) \reach_{f(\HH)} f(w)$. By induction hypothesis, we know that $u \reach_\HH w$. If $v \in H$, we obtain the expected result. If not, $v$ necessarily belongs to $\{x, y\}$. If, for instance, $v = x$, then $y \in H$. Thus $y$ is reachable from $u$ in $\HH$, and we conclude by $x \reacheq_\HH y$.
\end{proof}
It follows that the terminal \scc{}s of $\HH$ and $f(\HH)$ are in one-to-one correspondence. This property can be straightforwardly extended to the operation of merging several vertices of a same \scc{} simultaneously.

Using Propositions~\ref{prop:terminal_scc} and~\ref{prop:collapse}, we now sketch a method which computes the terminal \scc{}s in a directed hypergraph $\HH = (\VV,A)$. It performs several transformations on a hypergraph $\HH_\cur$ whose vertices are labelled by subsets of $\VV$:
\noindent
\begin{center}
\begin{tikzpicture}
\node[draw=black, text width=0.97\textwidth,text opacity=1] {
Starting from the hypergraph $\HH_\cur$ image of $\HH$ by the map $u \mapsto \{ u  \}$,
\begin{compactenum}[(i)]
\item\label{step:i} compute the terminal \scc{}s of the directed graph $\graph(\HH_\cur)$.
\item\label{step:ii} if one of them, say $C$, is not reduced to a singleton, replace $\HH_\cur$ by $f(\HH_\cur)$, where $f$ merges all the elements $U$ of $C$ into the vertex $\bigcup_{U \in C} U$. Then go back to Step~\eqref{step:i}.
\item\label{step:iii} otherwise, return the terminal \scc{}s of the directed graph $\graph(\HH_\cur)$.
\end{compactenum}};
\end{tikzpicture}
\end{center}
Each time the \emph{vertex merging step} (Step~\eqref{step:ii}) is executed, new arcs may appear in the directed graph $\graph(\HH_\cur)$. This case is illustrated in Figure~\ref{fig:merging}. In both sides, the arcs of $\graph(\HH_\cur)$ are depicted in solid, and the non-simple arcs of $\HH_\cur$ in dotted line. Note that the vertices of $\HH_\cur$ contain subsets of $\VV$, but enclosing braces are omitted for readability. Applying Step~\eqref{step:i} from vertex $u$ (left side) discovers a terminal \scc{} formed by $u$, $v$, and $w$ in the directed graph $\graph(\HH_\cur)$. At Step~\eqref{step:ii} (right side), the vertices are merged, and the hyperarc $a_4$ is transformed into two graph arcs leaving the new vertex $\{u,v,w\}$.

The termination of this method is ensured by the fact that the number of vertices in $\HH_\cur$ is strictly decreased each time Step~\eqref{step:ii} is applied. When the method is terminated, terminal \scc{}s of $\HH_\cur$ are all reduced to single vertices, each of them labelled by subsets of $\VV$. Propositions~\ref{prop:terminal_scc} and~\ref{prop:collapse} prove that these subsets are precisely the terminal \scc{}s of $\HH$. 

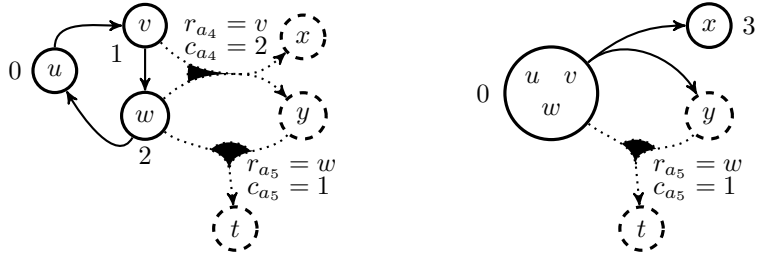
\begin{figure}[t]
\begin{center}
\begin{tikzpicture}[scale=0.6, vertex/.style={circle,draw=black,very thick,minimum size=2ex}, hyperedge/.style={draw=black,thick,dotted},simpleedge/.style={thick}]
\begin{scope}[scale=1, >=stealth']
\node [vertex,very thick] (u) at (-2,-1) {$u$} node[node distance=3.5ex,left of=u] {$0$};
\node [vertex,very thick] (v) at (0,0) {$v$} node[node distance=3.5ex,below left of=v] {$1$};
\node [vertex,very thick] (w) at (0,-2)  {$w$} node[node distance=3.5ex,below of=w] {$2$};
\node [vertex,dashed] (x) at (3.5,-0.25) {$x$};
\node [vertex,dashed] (y) at (3.5,-2) {$y$};
\node [vertex,dashed] (t) at (2,-4.5) {$t$};

\node at (1.75,-0.3) {$\begin{aligned} r_{a_4} & = v \\[-1.5ex] c_{a_4} & = 2 \end{aligned}$};
\node at (3.2,-3.4) {$\begin{aligned} r_{a_5} & = w \\[-1.5ex] c_{a_5} & = 1 \end{aligned}$};

\path[->] (u) edge[simpleedge,out=90,in=-180] (v);
\path[->] (v) edge[simpleedge,out=-90,in=90] (w);
\path[->] (w) edge[simpleedge,out=-120,in=-60] (u);
\hyperedge[0.6][$(hyper@tail)!0.6!(hyper@head)$]{v,w}{y,x};
\hyperedge{y,w}{t};
\end{scope}
\begin{scope}[scale=1, >=stealth', xshift=9cm]
\node [vertex,very thick,text width=0.7cm,text centered] (uvw) at (0,-1.5) {$u$ \hfill $v$\\ $w$} node[node distance=6ex,left of=uvw] {$0$};
\node [vertex] (x) at (3.5,0) {$x$} node[node distance=3.5ex,right of=x] {$3$};
\node [vertex,dashed] (y) at (3.5,-2) {$y$};
\node [vertex,dashed] (t) at (2,-4.5) {$t$};

\node at (3.2,-3.4) {$\begin{aligned} r_{a_5} & = w \\[-1.5ex] c_{a_5} & = 1 \end{aligned}$};
\path[->] (uvw) edge[simpleedge,out=40,in=-180] (x);
\path[->] (uvw) edge[simpleedge,out=40,in=120] (y);
\hyperedge[0.48]{y,uvw}{t};
\end{scope}
\end{tikzpicture}
\end{center}
\caption{A vertex merging step (the index of the visited vertices is given beside)}\label{fig:merging}
\end{figure}

\subsection{Optimized algorithm}\label{subsec:optimized_algorithm}

The sketch given in Section~\ref{subsec:maxscc_principle} is naturally not optimal (each vertex can be visited $O(\card{\VV})$ times). We propose to incorporate the vertex merging step directly into an algorithm determining the terminal \scc{}s in directed graphs, in order to gain efficiency. The resulting algorithm on directed hypergraphs is given in Figure~\ref{fig:maxscc}. We suppose that the directed hypergraph $\HH$ is provided with the lists $A_u$ of hyperarcs $a$ such that $u \in T(a)$, for each $u \in \VV$ (these lists can be built in linear time in a preprocessing step).

The algorithm consists of a main function \Call{TerminalScc}{} which initializes data, and then iteratively calls the function \Call{Visit}{} on the vertices which have not been visited yet. Following the sketch given in Section~\ref{subsec:maxscc_principle}, the function \Call{Visit}{$u$} repeats the following three tasks:
\begin{inparaenum}[(i)] 
\item it recursively searches a terminal \scc{} in the underlying directed graph $\graph(\HH_\cur)$, starting from the vertex $u$, 
\item once a terminal \scc{} is found, it performs a vertex merging step on it, 
\item and finally, it discovers the new graph arcs (if any) arising from the merging step.
\end{inparaenum}

Before discussing each of these three operations, we explain how the directed hypergraph $\HH_\cur$ is manipulated by the algorithm. Observe that the vertices of the hypergraph $\HH_\cur$ always form a partition of the initial set $\VV$ of vertices. Instead of referring to them as subsets of $\VV$, we use a union-find structure, which consists in three functions \Call{Find}{}, \Call{Merge}{}, and \Call{MakeSet}{} (see~\cite[Chapter~21]{Cormen01} for instance). A call to $\Call{Find}{u}$ returns, for each original vertex $u \in \VV$, the unique vertex of the hypergraph $\HH_\cur$ containing $u$. Two vertices $U$ and $V$ of $\HH_\cur$ can be merged by a call to $\Call{Merge}{U, V}$, which returns the new vertex. Finally, the ``singleton'' vertices $\{ u \}$ of the initial instance of the hypergraph $\HH_\cur$ are created by the function $\Call{MakeSet}{}$. In practice, each vertex of $\HH_\cur$ is encoded as a representative element $u \in \VV$, in which case the vertex corresponds to the subset $\{ v \in \VV \mid \Call{Find}{v} = u \}$. In other words, the hypergraph $\HH_\cur$ is precisely the image of $\HH$ by the function $\Call{Find}{}$.
To avoid confusion, we denote the vertices of the hypergraph $\HH$ by lower case letters, and the vertices of $\HH_\cur$ (and subsequently $\graph(\HH_\cur)$) by capital ones. By convention, if $u \in \VV$, $\Find{u}$ will correspond to the associated capital letter~$U$. 

\paragraph{Discovering terminal \scc{}s in the directed graph $\graph(\HH_{\cur})$}

This task is performed by the parts of the algorithm which are not shaded in gray. Similarly to Tarjan's algorithm~\cite{Tarjan72}, it uses a stack $S$ and two arrays indexed by vertices, $\tindex$ and $\troot$. The stack $S$ stores the vertices $U$ of $\graph(\HH_\cur)$ which are currently visited by $\Call{Visit}{}$. The array $\tindex$ tracks the order in which the vertices are visited, \ie\ $\tindex[U] < \tindex[V]$ if, and only if, $U$ has been visited by \Call{Visit}{} before $V$. 
The value $\troot[U]$ is used to determine the minimal index of the visited vertices which are reachable from $U$ (see Line~\lineref{scc:min}). A (non necessarily terminal) strongly connected component $C$ of $\graph(\HH_\cur)$ is discovered when a vertex $U$ satisfies $\troot[U] = \tindex[U]$ (Line~\lineref{scc:root}). Then $C$ consists of all the vertices stored in the stack $S$ above $U$. The vertex $U$ is the element of the \scc{} which has been visited first, and is called its \emph{root}. Once the visit of the \scc{} is terminated, its vertices are collected in a set $\incomp$ (Line~\lineref{scc:finished}).

Additionally, the algorithm uses an array $\ismax$ of booleans, allowing to track whether a \scc{} of $\graph(\HH_\cur)$ is terminal. A \scc{} is terminal if, and only if, its root $U$ satisfies $\ismax[U] = \True$. In particular, the boolean $\ismax[U]$ is set to $\False$ as soon as $U$ is connected to a vertex $W$ located in a distinct \scc{} (Line~\lineref{scc:not_max1}) or satisfying $\ismax[W] = \False$ (Line~\lineref{scc:not_max2}). 

\begin{figure}[t]
\begin{scriptsize}
\begin{minipage}[t]{0.43\textwidth}
\vspace{0pt}
\begin{algorithmic}[1]
\Function {TerminalScc}{$\HH = (\VV,A)$}
  \State \tikz[remember picture, baseline]{\coordinate (l6);}$n \gets 0$, $S \gets \emptystack$, $\incomp \gets \emptyset$ \tikz[remember picture, baseline]{\coordinate (t6);}
  \ForAll{$a \in A$} 
    \State $r_a \gets \Nil$, $c_a \gets 0$\tikz[remember picture, baseline]{\coordinate (r6);}
  \EndFor\tikz[remember picture, baseline]{\coordinate (b6);}
  \ForAll{$u \in \VV$} 
    \State $\tindex[u] \gets \Nil$
    \State $\troot[u] \gets \Nil$\tikz[remember picture, baseline]{\coordinate (t3);}
    \State \tikz[remember picture, baseline]{\coordinate (l3);}$F_u \gets \emptystack$, \Call{Makeset}{$u$}\tikz[remember picture, baseline]{\coordinate (b3);}\tikz[remember picture, baseline]{\coordinate (r3);}\label{scc:makeset} 
  \EndFor\label{scc:end_init}
  \ForAll{$u \in \VV$}\label{scc:begin_main_loop}
    \If{$\tindex[u] = \Nil$} 
      \State \Call{Visit}{$u$} \label{scc:visit_call}
    \EndIf
  \EndFor\label{scc:end_main_loop}
\EndFunction
\Statex
\Function {Visit}{$u$} 
  \State local $U \gets \Call{Find}{u}$\label{scc:find1}, local $F \gets \emptystack$\label{scc:begin}
  \State $\tindex[U] \gets n$, $\troot[U] \gets n$\label{scc:troot_def}
  \State $n \gets n+1$
  \State \tikz[remember picture, baseline]{\coordinate (t4);}$\ismax[U] \gets \True$
  \State push $U$ on the stack $S$
  \ForAll{$a \in A_u$}\label{scc:begin_node_loop}
    \sIf{$\card{T(a)} = 1$} push $a$ on $F$\label{scc:simple_edge}
    \fElse
      \sIf{$r_a = \Nil$} $r_a \gets u$\label{scc:root_def}\tikz[remember picture, baseline]{\coordinate (t);} \tikz[remember picture, baseline]{\coordinate (r);}
      \State local $R_a \gets \Call{Find}{r_a}$\label{scc:find2}
      \If{$R_a$ appears in $S$}\label{scc:root_reach}
	  \State $c_a \gets c_a + 1$\label{scc:counter_increment}
	  \If{$c_a = \card{T(a)}$}\label{scc:counter_reach}
	    \State push $a$ on stack $F_{R_a}$\label{scc:stack_edge}
	  \EndIf
      \EndIf\tikz[remember picture, baseline]{\coordinate (b);}\label{scc:end_aux}
    \EndIf
  \EndFor\label{scc:end_node_loop}\tikz[remember picture, baseline]{\coordinate (l);}
\algstore{scc_break}
\end{algorithmic}
\end{minipage}
\hfill
\begin{minipage}[t]{0.55\textwidth}
\vspace{0pt}
\begin{algorithmic}[1]
\algrestore{scc_break}
  \While{$F$ is not empty} \label{scc:begin_edge_loop}\tikz[remember picture]{\node[anchor=south] (goto2) {};}\tikz[remember picture, baseline]{\coordinate (t5);}
    \State pop $a$ from $F$
    \ForAll{$w \in H(a)$} \label{scc:begin_edge_loop2}
      \State local $W \gets \Call{Find}{w}$\label{scc:find3}
      \sIf{$\tindex[W] = \Nil$} \Call{Visit}{$w$}
      \If{$W \in \incomp$} \label{scc:membership}
	\State $\ismax[U] \gets \False$ \label{scc:not_max1}
      \Else
	\State $\troot[U] \gets \min(\troot[U],\troot[W])$ \label{scc:min}
	\State $\ismax[U] \gets \ismax[U] \And \ismax[W]$ \label{scc:not_max2} \tikz[remember picture, baseline]{\coordinate (r5);}
      \EndIf
    \EndFor \label{scc:end_edge_loop2}
  \EndWhile \label{scc:end_edge_loop}\tikz[remember picture, baseline]{\coordinate (b5);}
  \If{$\troot[U] = \tindex[U]$} \label{scc:root} \label{scc:begin2}
    \If{$\ismax[U] = \True$} 
    \Statex \qquad \Lcomment{a terminal \scc{} is discovered} \tikz[remember picture, baseline]{\coordinate (t2);}
      \State \tikz[remember picture, baseline]{\coordinate (l2);}local $i \gets \tindex[U]$\label{scc:begin_node_merging}
      \State pop each $a$ from $F_U$ and push it on $F$\label{scc:push_on_fprime1}
      \State pop $V$ from $S$
      \While{$\tindex[V] > i$} \label{scc:begin_node_merging_loop}
	\State pop each $a$ from $F_V$ and push it on $F$ \label{scc:push_on_fprime2}\tikz[remember picture, baseline]{\coordinate (r2);}
	\State $U \gets \Call{Merge}{U, V}$\label{scc:merge}
	\State pop $V$ from $S$
      \EndWhile \label{scc:end_node_merging_loop}\label{scc:end_node_merging}
      \State $\tindex[U] \gets i$, push $U$ on $S$\label{scc:index_redef}
      \sIf{$F$ is not empty} go to Line~\lineref{scc:begin_edge_loop}\tikz[remember picture, baseline]{\coordinate (b2);}\label{scc:end_node_merging2}\label{scc:goto} \tikz[remember picture]{\node[anchor=south] (goto1) {};}
    \EndIf
    \Repeat \label{scc:begin_non_max_scc_loop}
      \State pop $V$ from $S$, add $V$ to $\incomp$ \label{scc:finished}
    \Until{$\tindex[V] = \tindex[U]$} \label{scc:end_non_max_scc_loop}
  \EndIf \label{scc:end2}
\EndFunction\tikz[remember picture, baseline]{\coordinate (b4);}\label{scc:end}
\end{algorithmic}
\end{minipage}
\end{scriptsize}
\begin{tikzpicture}[remember picture,overlay]
\path (t) ++ (0ex,2ex) coordinate (t);
\path (b) ++ (0ex,-0.5ex) coordinate (b);
\path (l) ++ (-0.5ex,0ex) coordinate (l);
\path (r) ++ (0.5ex,0ex) coordinate (r);
\path let \p1 = (l), \p2 = (t) in coordinate (lt) at (\x1,\y2);
\path let \p1 = (l), \p2 = (b) in coordinate (lb) at (\x1,\y2);
\path let \p1 = (r), \p2 = (t) in coordinate (rt) at (\x1,\y2);
\path let \p1 = (r), \p2 = (b) in coordinate (rb) at (\x1,\y2);
\filldraw[draw=none, style=very nearly transparent, fill=gray!70!black] (lt) -- (rt) -- (rb) -- (lb) -- cycle;
\draw[draw=black] (rt) ++ (-1ex,0ex) -- (rt) -- (rb) -- ++ (-1ex,0ex);
\node[anchor=north,text width=0.4\textwidth] at ($(rb)$) {auxiliary data update step};
\path (t2) ++ (0ex,-0.5ex) coordinate (t2);
\path (b2) ++ (0ex,-0.5ex) coordinate (b2);
\path (l2) ++ (-0.5ex,0ex) coordinate (l2);
\path (r2) ++ (0.5ex,0ex) coordinate (r2);
\path let \p1 = (l2), \p2 = (t2) in coordinate (lt2) at (\x1,\y2);
\path let \p1 = (l2), \p2 = (b2) in coordinate (lb2) at (\x1,\y2);
\path let \p1 = (r2), \p2 = (t2) in coordinate (rt2) at (\x1,\y2);
\path let \p1 = (r2), \p2 = (b2) in coordinate (rb2) at (\x1,\y2);
\filldraw[draw=none, style=very nearly transparent, fill=gray!70!black] (lt2) -- (rt2) -- (rb2) -- (lb2) -- cycle;
\draw[draw=black] (rt2) ++ (-1ex,0ex) -- (rt2) -- (rb2) -- ++ (-1ex,0ex);
\node[anchor=base west,text width=0.12\textwidth] at ($(rt2)!0.5!(rb2)$) {vertex merging step};
\path (t3) ++ (0ex,-0.5ex) coordinate (t3);
\path (b3) ++ (0ex,-0.5ex) coordinate (b3);
\path (l3) ++ (-0.5ex,0ex) coordinate (l3);
\path (r3) ++ (0.5ex,0ex) coordinate (r3);
\path let \p1 = (l3), \p3 = (t3) in coordinate (lt3) at (\x1,\y3);
\path let \p1 = (l3), \p3 = (b3) in coordinate (lb3) at (\x1,\y3);
\path let \p1 = (r3), \p3 = (t3) in coordinate (rt3) at (\x1,\y3);
\path let \p1 = (r3), \p3 = (b3) in coordinate (rb3) at (\x1,\y3);
\filldraw[draw=none, style=very nearly transparent, fill=gray!70!black] (lt3) -- (rt3) -- (rb3) -- (lb3) -- cycle;
\path (t6) ++ (0ex,-0.5ex) coordinate (t6);
\path (b6) ++ (0ex,-0.5ex) coordinate (b6);
\path (l6) ++ (-0.5ex,0ex) coordinate (l6);
\path (r6) ++ (0.5ex,0ex) coordinate (r6);
\path let \p1 = (l6), \p3 = (t6) in coordinate (lt6) at (\x1,\y3);
\path let \p1 = (l6), \p3 = (b6) in coordinate (lb6) at (\x1,\y3);
\path let \p1 = (r6), \p3 = (t6) in coordinate (rt6) at (\x1,\y3);
\path let \p1 = (r6), \p3 = (b6) in coordinate (rb6) at (\x1,\y3);
\filldraw[draw=none, style=very nearly transparent, fill=gray!70!black] (lt6) -- (rt6) -- (rb6) -- (lb6) -- cycle;
\end{tikzpicture}
\caption{Computing the terminal \scc{}s in directed hypergraphs}\label{fig:maxscc}
\end{figure}
  
\paragraph{Vertex merging step} This step is performed from Lines~\lineref{scc:begin_node_merging} to~\lineref{scc:end_node_merging2}, when it is discovered that the vertex $U = \Find{u}$ is the root of a terminal \scc{} in the digraph $\graph(\HH_\cur)$. All vertices $V$ which have been collected in that \scc{} are merged to $U$ (Line~\lineref{scc:merge}). Let $\HH_\new$ be the resulting hypergraph. 

At Line~\lineref{scc:end_node_merging2}, the stack $F$ is expected to contain the new arcs of the directed graph $\graph(\HH_\new)$ leaving the newly ``big'' vertex $U$ (this point will be explained in the next paragraph). If $F$ is empty, the singleton $\{U\}$ constitutes a terminal \scc{} of $\graph(\HH_\new)$, hence also of $\HH_\new$ (Proposition~\ref{prop:terminal_scc}). Otherwise, we go back to Line~\lineref{scc:begin_edge_loop} to discover terminal \scc{}s from the new vertex $U$ in the digraph $\graph(\HH_\new)$. 

\paragraph{Discovering the new graph arcs} In this paragraph, we explain informally how the new graph arcs arising after a vertex merging step can be efficiently discovered without examining all the hyperarcs. The formal proof of this technique is provided in Appendix~\ref{sec:correctness_proof}.

During the execution of \Call{Visit}{$u$}, the local stack $F$ is used to collect the hyperarcs which represent arcs leaving the vertex $\Find{u}$ in $\graph(\HH_\cur)$. Initially, when \Call{Visit}{$u$} is called, the vertex \Call{Find}{$u$} is still equal to $u$. Then, the loop from Lines~\lineref{scc:begin_node_loop} to~\lineref{scc:end_node_loop} iterates over the set $A_u$ of the hyperarcs $a \in A$ such that $u \in T(a)$. At the end of the loop, it can be verified that $F$ is indeed filled with all the simple hyperarcs leaving $u = \Find{u}$ in $\HH_\cur$, as expected (see Line~\lineref{scc:simple_edge}).

The main difficulty is to collect in $F$ the arcs which are added to the digraph $\graph(\HH_\cur)$ after a vertex merging step. To this aim, each non-simple hyperarc $a \in A$ is provided with two auxiliary data: 
\begin{itemize}[\textbullet]
\item a vertex $r_a$, called the \emph{root} of the hyperarc $a$, which is defined as the first vertex of the tail $T(a)$ to be visited by a call to \Call{Visit}{},
\item a counter $c_a \geq 0$, which determines the number of vertices $x \in T(a)$ which have been visited and such that $\Call{Find}{x}$ is reachable from $\Call{Find}{r_a}$ in the current digraph $\graph(\HH_\cur)$.
\end{itemize}
These auxiliary data are maintained in the \emph{auxiliary data update step}, located from Lines~\lineref{scc:root_def} to~\lineref{scc:end_aux}. Initially, the root $r_a$ of any hyperarc $a$ is set to the special value $\Nil$. The first time a vertex $u$ such that $a \in A_u$ is visited, $r_a$ is assigned to $u$ (Line~\lineref{scc:root_def}). Besides, at the call to \Call{Visit}{$u$}, the counter $c_a$ of each non-simple hyperarc $a \in A_u$ is incremented, but only when $R_a = \Find{r_a}$ belongs to the stack $S$ (Line~\lineref{scc:counter_increment}). This is indeed a necessary and sufficient condition to the fact that $\Find{u}$ is reachable from $\Find{r_a}$ in the digraph $\graph(\HH_\cur)$ (see Invariant~\ref{inv:call_to_visit3} in Appendix~\ref{sec:correctness_proof}). 

It follows from these invariants that, when the counter $c_a$ reaches the threshold value $\card{T(a)}$, all the vertices $X = \Find{x}$, for $x \in T(a)$, are reachable from $R_a$ in the digraph $\graph(\HH_\cur)$. Now suppose that, later, it is discovered that $R_a$ belongs to a terminal strongly connected component $C$ of $\graph(\HH_\cur)$. Then the aforementioned vertices $X$ must all stand in the component $C$ (since it is terminal). Therefore, when the vertex merging step is applied on this \scc{}, the vertices $X$ are merged into a single vertex $U$. Hence, the hyperarc $a$ necessarily generates new simple arcs leaving $U$ in the new version of the digraph $\graph(\HH_\cur)$. Let us verify that in this case, $a$ is correctly placed into $F$ by our algorithm. As soon as $c_a$ reaches the value $\card{T(a)}$, the hyperarc $a$ is placed into a temporary stack $F_{R_a}$ associated to the vertex $R_a$ (Line~\lineref{scc:stack_edge}). This stack is then emptied into $F$ during the vertex merging step, at Lines~\lineref{scc:push_on_fprime1} or~\lineref{scc:push_on_fprime2}. 

\begin{example}
In the left side of Figure~\ref{fig:merging}, the execution of the loop from Lines~\lineref{scc:begin_node_loop} to~\lineref{scc:end_node_loop} during the call to \Call{Visit}{$v$} sets the root of the hyperarc $a_4$ to the vertex $v$, and $c_{a_4}$ to $1$. Then, during \Call{Visit}{$w$}, $c_{a_4}$ is incremented to $2 = \card{T(a_4)}$. The hyperarc $a_4$ is therefore pushed on the stack $F_{v}$ (because $R_{a_4} = \Find{r_{a_4}} = \Find{v} = v$). Once it is discovered that $u$, $v$, and $w$ form a terminal \scc{} of $\graph(\HH_\cur)$, $a_4$ is collected into $F$ during the merging step. It then allows to visit the vertices $x$ and $y$ from the new vertex (rightmost hypergraph). A fully detailed execution trace is provided in Appendix~\ref{sec:execution_trace} below.
\end{example}

\paragraph{Correctness and complexity} 

For the sake of simplicity, we have not included in \Call{TerminalScc}{} the step returning the terminal \scc{}s. However, they can be easily built by examining each vertex (hence in time $O(\card{\VV})$), as shown below:
\begin{theorem}\label{th:correctness}
Let $\HH = (\VV,A)$ be a directed hypergraph. After the execution of \Call{TerminalScc}{$\HH$}, the terminal strongly connected components of $\HH$ are precisely the sets $C_U = \{ v \in \VV \mid \Call{Find}{v} = U \text{ and } \ismax[U] = \True \}$.
\end{theorem}
The proof of Theorem~\ref{th:correctness}, which is too long to be included here, is provided in Appendix~\ref{sec:correctness_proof}. It relies on successive transformations of intermediary algorithms to \Call{TerminalScc}{}.

When using disjoint-set forests with union by rank and path compression as union-find structure (see~\cite[Chapter 21]{Cormen01}), the time complexity of any sequence of $p$ operations \Call{MakeSet}{}, \Call{Find}{}, or \Call{Merge}{} is known to be $O(p \times \alpha(\card{\VV}))$, where $\alpha$ is the very slowly growing inverse of the Ackermann function. The following result states that the algorithm \Call{TerminalScc}{} is also almost linear time:
\begin{theorem}\label{th:complexity}
Let $\HH = (\VV,A)$ be a directed hypergraph. Then the algorithm \Call{TerminalScc}{$\HH$} terminates in time $O(\size(\HH) \times \alpha(\card{\VV}))$, and has linear space complexity. 
\end{theorem}

\begin{proof}
The analysis of the time complexity \Call{TerminalScc}{} depends on the kind of the instructions. We distinguish:
\begin{inparaenum}[(i)]
\item the operations on the global stacks $F_u$ and on the local stacks $F$,
\item the call to the functions \Call{Find}{}, \Call{Merge}{}, and \Call{MakeSet}{},
\item and the other operations, referred to as \emph{usual operations} (by extension, their time complexity will be referred to as \emph{usual complexity}). 
\end{inparaenum}
The complexity of each kind of operations is respectively described in the following three paragraphs.

Each operation on a stack (pop or push) is performed in $O(1)$. A hyperarc $a$ is pushed on a stack of the form $F_u$ at most once during the whole execution of \Call{TerminalScc}{} (when counter $c_a$ reaches the value $\card{T(a)}$). Once it is popped from it, it will never be pushed on a stack of the form $F_v$ again. Similarly, a hyperarc is pushed on a local stack $F$ at most once, and after it is popped from it, it will never be pushed on any local stack $F'$ in the following states. Therefore, the total number of stack operations on the local and global stacks $F$ and $F_u$ is bounded by $4 \card{A}$. It follows that the corresponding complexity is bounded by $O(\size(\HH))$. The same argument proves that the total number of iterations of the loop from Lines~\lineref{scc:begin_edge_loop2} to~\lineref{scc:end_edge_loop2} occurring in a complete execution of \Call{TerminalScc}{} is bounded by $\sum_{a \in A} \card{H(a)}$.

During the execution of \Call{TerminalScc}{}, the function \Find{} is called exactly $\card{\VV}$ times at Line~\lineref{scc:find1}, at most $\sum_{u \in \VV} \card{A_u} = \sum_{a \in A} \card{T(a)}$ times at Line~\lineref{scc:find2}, and at most $\sum_{a \in A} \card{H(a)}$ times at Line~\lineref{scc:find3} (see above). Hence it is called at most $\size(\HH)$ times. The function \Call{Merge}{} is always called to merge two distinct vertices. Let $C_1,\dots,C_p$ ($p \leq \card{\VV}$) be the equivalence classes formed by the elements of $\VV$ at the end of the execution of \Call{TerminalScc}{}. Then \Call{Merge}{} is called at most $\sum_{i = 1}^p (\card{C_i}-1)$. Since $\sum_i \card{C_i} = \card{\VV}$, \Call{Merge}{} is executed at most $\card{\VV}-1$ times. Finally, \Call{MakeSet}{} is called exactly $\card{\VV}$ times. It follows that the total time complexity of the operations \Call{MakeSet}{}, \Find{} and \Call{Merge}{} is $O(\size(\HH) \times \alpha(\card{\VV}))$.

The analysis of the usual operations is split into several parts:
\begin{itemize}
\item the usual complexity of \Call{TerminalScc}{} without the calls to the function \Call{Visit}{} is clearly $O(\card{\VV} + \card{A})$.
\item during the execution of \Call{Visit}{$u$}, the usual complexity of the block from Lines~\lineref{scc:begin} to~\lineref{scc:end_node_loop} is $O(1) + O(\card{A_u})$. Indeed, we suppose that the test at Line~\lineref{scc:root_reach} can be performed in $O(1)$ by assuming that the stack $S$ is provided with an auxiliary array of booleans which determines, for each element of $\VV$, whether it is stored in $S$ (obviously, the push and pop operations are still in $O(1)$ under this assumption). Then the total usual complexity between Lines~\lineref{scc:begin} and~\lineref{scc:end_node_loop} is $O(\size(\HH))$ for a complete execution of \Call{TerminalScc}{}.
\item the usual complexity of the loop body from Lines~\lineref{scc:begin_edge_loop2} to~\lineref{scc:end_edge_loop2}, without the recursive calls to \Call{Visit}{}, is clearly $O(1)$ (the membership test at Line~\lineref{scc:membership} is supposed to be in $O(1)$, encoding the set $\incomp$ as an array of $\card{\VV}$ booleans). This inner loop is iterated $\card{H(a)}$ times during each iteration of the outer loop from Lines~\lineref{scc:begin_edge_loop} to~\lineref{scc:end_edge_loop}. Since a hyperarc is placed in a local stack $F$ at most once, the total usual complexity of the loop from Lines~\lineref{scc:begin_edge_loop} to~\lineref{scc:end_edge_loop} (without the recursive calls to \Call{Visit}{}) is bounded by $O(\size(\HH))$.
\item the usual complexity of the loop between Lines~\lineref{scc:begin_node_merging_loop} and~\lineref{scc:end_node_merging_loop} for a complete execution of \Call{TerminalScc}{} is $O(\card{\VV})$, since in total, it is iterated exactly the number of times the function \Call{Merge}{} is called.
\item the usual complexity of the loop between Lines~\lineref{scc:begin_non_max_scc_loop} and~\lineref{scc:end_non_max_scc_loop} for a whole execution of \Call{TerminalScc}{} is $O(\card{\VV})$, because a given element is placed at most once into the set $\incomp$ (adding an element in $\incomp$ is in $O(1)$).
\item if the previous two loops are not considered, less than $10$ usual operations are executed in the block from Lines~\lineref{scc:begin2} to~\lineref{scc:end}, all of complexity $O(1)$. The execution of this block either follows a call to \Call{Visit}{} or the execution of the goto statement (at Line~\lineref{scc:goto}). The latter is executed only if the stack $F$ is not empty. Since each hyperarc can be pushed on a local stack $F$ and then popped from it only once, it happens at most $\card{A}$ times during the whole execution of \Call{TerminalScc}{}. It follows that the usual complexity of the block from Lines~\lineref{scc:begin2} to~\lineref{scc:end} is $O(\card{\VV} + \card{A})$ in total (excluding the loops previously discussed).
\end{itemize}

Summing all the complexities above proves that the time complexity of \Call{TerminalScc}{} is $O(\size(\HH) \times \alpha(\card{\VV}))$. The space complexity is obviously linear in $\size(\HH)$.
\end{proof}

An implementation is provided in the library~\tplib{}~\cite{tplib}, in the module \texttt{Hypergraph}.\footnote{The module can be used independently of the rest of the library. Note that in the source code, terminal \scc{}s are referred to as \emph{maximal} \scc{}s.}

\begin{remark}
The algorithm \Call{TerminalScc}{} is not able to determine all strongly connected components in directed hypergraphs. Consider the following example:
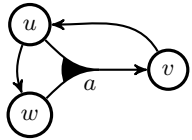
\begin{center}
\begin{small}
\begin{tikzpicture}[>=stealth',scale=0.8,vertex/.style={circle,draw=black,very thick,minimum size=2ex}, hyperedge/.style={draw=black,thick}, simpleedge/.style={draw=black,thick}]
\node [vertex] (u) at (0, 0) {$u$};
\node [vertex] (t) at (0, -1.5) {$w$};

\node [vertex] (x) at (2.3, -0.75) {$v$};

\path[->] (u) edge[simpleedge,out=-110,in=110] (t);

\path[->] (x) edge[simpleedge,out=120,in=0] (u);

\hyperedge[0.6]{u,t}{x};
\node at (1,-1) {$a$};
\end{tikzpicture}
\end{small}
\end{center}
Our algorithm determines the unique terminal \scc{}, which is reduced to the vertex $w$. However, the non-terminal \scc{} formed by $u$ and $v$ is not discovered. Indeed, the non-simple hyperarc $a$, which allows to reach $v$ from $u$, cannot be transformed into a simple arc, since $u$ and $v$ do not belong to a same \scc{} of the underlying digraph.
\end{remark}

\subsection{Determining other properties in almost linear time, and applications}\label{subsec:other_properties}

Some properties can be directly determined from the terminal \scc{}s. Indeed, a directed hypergraph $\HH$ admits a sink (\ie\ a vertex reachable from all vertices) if, and only if, it contains a unique terminal \scc{}. Besides, strong connectivity amounts to the existence of a terminal \scc{} containing all the vertices.
\begin{corollary}\label{cor:strongly_connectivity_and_sink}
Given a directed hypergraph $\HH$, the following problems can be solved in almost linear time in $\size(\HH)$: 
\begin{inparaenum}[(i)]
\item\label{item:i} is there a sink in $\HH$?	
\item\label{item:ii} is $\HH$ strongly connected?
\end{inparaenum}
\end{corollary}

We now discuss some applications of these results.

\paragraph{Tropical geometry}

Tropical polyhedra are the analogues of convex polyhedra in \emph{tropical algebra}, \ie\ the semiring $\mathbb{R} \cup \{-\infty\}$ endowed with the operations $\max$ and $+$ as addition and multiplication. A \emph{tropical polyhedron} is the set of the solutions $x \in (\mathbb{R} \cup \{-\infty\})^n$ of finitely many tropical affine inequalities, of the form:
\[
\max(a_0, a_1 + x_1, \dots, a_n + x_n) \leq \max(b_0, b_1 + x_1, \dots, b_n + x_n) \enspace, 
\]
where $a_i, b_i \in \mathbb{R} \cup \{-\infty\}$. 

Analogously to classical convex polyhedra, any tropical polyhedron can be equivalently expressed as the convex hull (in the tropical sense) of a set of vertices and extreme rays. This yields the problem of computing the vertices of a tropical polyhedron, which can be seen as the ``tropical counterpart'' of the well-studied vertex enumeration problem in computational geometry. This problem has various applications in computer science and control theory, among others, in the analysis of discrete event systems~\cite{katz05}, software verification~\cite{AllamigeonGaubertGoubaultSAS08}, and verification of real-time systems~\cite{LuJLAP2011}.

Directed hypergraphs and their terminal \scc{}s arise in the characterization of the vertices of a tropical polyhedron. In a joint work of the author with Gaubert and Goubault~\cite{AllamigeonGaubertGoubaultDCG2013}, it has been proved that a point $x \in (\mathbb{R} \cup \{-\infty\})^n$ of a tropical polyhedron $\mathcal{P}$ is a vertex if, and only if, a certain directed hypergraph associated to $x$ and built from the inequalities defining $\mathcal{P}$, admits a sink. 

This combinatorial criterion plays a crucial role in the tropical vertex enumeration problem. It is indeed involved in an algorithm called \emph{tropical double description method}~\cite{AllamigeonGaubertGoubaultDCG2013,AllamigeonGaubertGoubaultSTACS10}, in order to eliminate points which are not vertices, among a set of candidates. This set can be very large (exponential in the dimension $n$), so the efficiency of the elimination step is critical. The almost linear time algorithm \Call{TerminalScc}{} consequently leads to a significant improvement over the state-of-the-art, both in theory and in practice (see~\cite[Section~6]{AllamigeonGaubertGoubaultDCG2013}). It also allows to show the surprising result that it is easier to determine whether a point is a vertex in a tropical polyhedron than in a classical one (if $p$ is the number of inequalities defining the polyhedron, the latter problem can be solved in $O(n^2 p)$ while the former in $O(n p \alpha(n))$). 

\paragraph{Nonlinear spectral theory}

Problem~\eqref{item:ii} appears in a generalization of the Perron-Frobenius theorem to homogeneous and monotone functions studied by Gaubert and Gunawardena in~\cite{GaubertGunawardena04}. Recall that a function $f : (\mathbb{R}^*_+)^n \mapsto (\mathbb{R}^*_+)^n$ is said to be \emph{monotone} when $f(x) \leq f(y)$ for any $x, y \in (\mathbb{R}^*_+)^n$ such that $x \leq y$ (the relation $\leq$ being understood entrywise), and that it is \emph{homogeneous} when $f(\lambda x) = \lambda f(x)$ for all $\lambda \in \mathbb{R}^*_+$ and $x \in (\mathbb{R}^*_+)^n$. A central problem is to give conditions under which $f$ admits an eigenvector in the cone $(\mathbb{R}^*_+)^n$, \ie\ a vector $x \in (\mathbb{R}^*_+)^n$ such that $f(x) = \lambda x$ for some $\lambda > 0$. Gaubert and Gunawardena establish a sufficient combinatorial condition~\cite[Theorem~6]{GaubertGunawardena04} 
expressed as the strong connectivity of a directed graph obtained as the limit of a sequence of graphs. This sequence is identical to the one arising during the execution of the method sketched in Section~\ref{subsec:maxscc_principle}. It follows that the sufficient condition is equivalent to the strong connectivity of a directed hypergraph $\HH(f)$ constructed from $f$. The hypergraph $\HH(f)$ consists of the vertices $1, \dots, n$ and the hyperarcs $(I, \{j\})$ such that $\lim_{\mu \rightarrow +\infty} f_j(\mu e_I) = +\infty$ ($e_I$ denotes the vector whose $i$-th entry is equal to $1$ when $i \in I$, and $0$ otherwise).

\paragraph{Horn propositional logic} As mentioned in Section~\ref{sec:introduction}, directed hypergraphs can be used to encode Horn formulas. Recall that a \emph{Horn formula} $F$ over the propositional variables $X_1, \dots, X_n$ is a conjunction of \emph{Horn clauses}, \ie\ 
\begin{inparaenum}[(i)]
\item either an implication $X_{i_1} \wedge \dots \wedge X_{i_p} \Rightarrow X_i$,
\item or a fact $X_i$,
\item or a goal $\neg X_{i_1} \vee \dots \vee \neg X_{i_p}$.
\end{inparaenum}
Given a propositional formula $F$, an assignment $\sigma : \{ X_1, \dots, X_n \} \rightarrow \{\True, \False\}$ is a \emph{model} of $F$ if replacing each $X_i$ by its associated truth value $\sigma(X_i)$ yields a true assertion. If $F_1, F_2$ are two propositional formulas, $F_1$ is said to \emph{entail} $F_2$, which is denoted by $F_1 \models F_2$, if every model of $F_1$ is a model of $F_2$. 

A directed hypergraph $\HH(F)$ can be associated to any Horn formula $F$ so as to decide entailment of implications over its variables. We use the construction developed by Ausiello and Italiano~\cite{Ausiello91}. The hypergraph $\HH(F)$ consists of the vertices $\vtrue, \vfalse, 1, \dots, n$ and the following hyperarcs: $(\{i_1, \dots, i_p\}, \{i\})$ for every implication $X_{i_1} \wedge \dots \wedge X_{i_p} \Rightarrow X_i$ in $F$, $(\{\vtrue\}, \{i\})$ for every fact $X_i$, $(\{i_1, \dots, i_p\}, \{\vfalse\})$ for every goal $\neg X_{i_1} \vee \dots \vee \neg X_{i_p}$, $(\{\vfalse\}, \{1, \dots, n\})$, and $(\{i\},\{\vtrue\})$ for all $i = 1, \dots, n$.
Observe that the size of $\HH(F)$ is linear in the size of the formula $F$, \ie\ the number of atoms in its clauses (without loss of generality, it is assumed that every variable occurs in $F$). 
\begin{lemma}\label{lemma:horn}
Let $F$ be a Horn formula over the variables $X_1, \dots, X_n$. Then $F \models X_i \Rightarrow X_j$ if, and only if, $j$ is reachable from $i$ in the directed hypergraph $\HH(F)$.
\end{lemma}

\begin{proof}
The ``if'' part can be shown by induction. If $i = j$, this is obvious. Otherwise, there exists a hyperarc $(T,H)$ such that $j \in H$ and every element of $T$ is reachable from $i$. Three cases can be distinguished:
\begin{itemize}
\item if $T$ is not equal to $\{\vtrue\}$ or $\{\vfalse\}$, then by construction, $F$ contains the implication $\wedge_{k \in T} X_k \Rightarrow X_j$. For all $k \in T$, $F \models X_i \Rightarrow X_k$ by induction, so that $F \models X_i \Rightarrow X_j$.
\item if $T = \{\vtrue\}$, then $X_j$ is a fact, and $F \models X_i \Rightarrow X_j$ trivially holds.
\item if $T$ is reduced to $\{\vfalse\}$, there is a goal $\neg X_{k_1} \vee \dots \vee \neg X_{k_p}$ in $F$ such that each $k_l$ is reachable from $i$, hence $F \models X_i \Rightarrow X_{k_l}$. As $F$ also entails the implication $X_{k_1} \wedge \dots \wedge X_{k_p} \Rightarrow X_j$, we conclude that $F \models X_i \Rightarrow X_j$.
\end{itemize}
For the ``only if'' part, let $R$ be the set of reachable vertices from $i$ in $\HH(F)$, and assume that $j \not \in R$. Let $\sigma$ be the assignment defined by $\sigma(X_k) = \True$ if $k \in R$, $\False$ otherwise. We claim that $\sigma$ models $F$. Consider an implication $X_{k_1} \wedge \dots \wedge X_{k_p} \Rightarrow X_k$ in $F$. If $\sigma(X_{k_l}) = \True$ for all $l = 1, \dots, p$, then $k$ is reachable from $i$ in $\HH(F)$, hence $\sigma(X_k) = \True$, and the implication is valid on $\sigma$. Similarly, for each fact $X_k$ in $F$, $k$ obviously belongs to $R$, which ensures that $\sigma(X_k) = \True$. Finally, if $F$ contains a goal $\neg X_{k_1} \vee \dots \vee \neg X_{k_p}$ such that $\sigma(X_{k_l}) = \True$ for all $l = 1, \dots, p$, then every vertex of the hypergraph is reachable from $i$ (through the vertex $\vfalse$), which is impossible ($j \not \in R$). This completes the proof.
\end{proof}
Corollary~\ref{cor:strongly_connectivity_and_sink} and Lemma~\ref{lemma:horn} consequently prove that the two following decision problems over Horn formulas can be solved in almost linear time: 
\begin{inparaenum}[(i)]
\item whether a variable of a Horn formula is implied by all the others,
\item whether all variables of a Horn formula are equivalent.
\end{inparaenum}

\section{Contributions on the complexity of computing all \scc{}s}\label{sec:combinatorics}

\subsection{A lower bound on the size of the transitive reduction of the reachability relation}\label{subsec:transitive_reduction}

Given a directed graph or a directed hypergraph, the reachability relation can be represented by the set of the couples $(x,y)$ such that $x$ reaches $y$. This is however a particularly redundant representation because of transitivity. In order to get a better idea of the intrinsic complexity of the reachability relation, we should rather consider transitive reductions, which are defined as minimal binary relations having the same transitive closure. 

In directed graphs, Aho~\etal\  have shown in~\cite{AhoGareyUllmanSICOMP72} that all transitive reductions of the reachability relation have the same size (the size of a binary relation $\RR$ is the number of couples $(x,y)$ such that $x \mathbin{\RR} y$). This size is bounded by the size of the digraph. Furthermore, a canonical transitive reduction can be defined by choosing a total ordering over the vertices.

In directed hypergraphs, the existence of a canonical transitive reduction of the reachability relation can be similarly established, because reachability is still reflexive and transitive.\footnote{Any finite reflexive and transitive relation $\RR$ can be seen as the reachability relation of a directed graph $G$, whose arcs are the couples $(x,y)$ such that $x \mathbin{\RR} y$, $x \neq y$. Then the transitive reduction of $\RR$ is defined as in~\cite{AhoGareyUllmanSICOMP72}.} However, we are going to show that its size is superlinear in $\size(\HH)$ for some directed hypergraphs $\HH$.

These hypergraphs arise from the subset partial order. More specifically, given a family $\FF$ of distinct sets over a finite domain $D$, the partial order induced by the relation $\subseteq$ on $\FF$ is called \emph{the subset partial order} over $\FF$. Without loss of generality, we assume that every set $S$ of $\FF$ satisfies $\card{S} > 1$ (up to adding two fixed elements $x,y \not \in D$ to all sets, which does not change the partial order over $\FF$). From this family, we build a corresponding directed hypergraph $\HH(\FF,D)$. Each of its vertices is either associated to a set $S \in \FF$ or to a domain element $x \in D$, and is denoted by $v[S]$ or $v[x]$ respectively. Besides, each set $S$ is associated to two hyperarcs $a[S]$ and $a'[S]$. The hyperarc $a[S]$ leaves the singleton $\{v[S]\}$ and enters the set of the vertices $v[x]$ such that $x \in S$. The hyperarc $a'[S]$ is defined inversely, leaving the latter set and entering $\{v[S]\}$. An example is given in Figure~\ref{fig:subset_hypergraph}.

\begin{figure}
\begin{center}
\begin{tikzpicture}[scale=0.9,>=stealth',elt/.style={circle,inner sep=1pt,draw=black,thick,font=\tiny}, set/.style={rectangle, minimum height=0.75cm,minimum width=0.75cm,inner sep=1pt,draw=black,thick,font=\tiny}, hyperedge/.style={draw=black,thick}, simpleedge/.style={black,thick},lbl/.style={font=\tiny}]
\node [elt] (x1) at (0,1.5) {$v[x_1]$};
\node [elt] (x2) at (2,1.5) {$v[x_2]$};
\node [elt] (x3) at (2,0)  {$v[x_3]$};
\node [elt] (x4) at (0,0)  {$v[x_4]$};

\node [set] (s1) at (-2,0.75) {$v[S_1]$};
\node [set] (s2) at (4,0.75) {$v[S_2]$};
\node [set] (s3) at (1,3.25) {$v[S_3]$};

\hyperedgewithangles[0.15][$(-1.6,-0.5)$]{x2/-100,x4/-145,x1/-110}{165}{s1.south/-80};
\path (-1.6,-0.5) node[below,lbl] {$a'[S_1]$};
\hyperedgewithangles[0.3][$(-1.2,0.75)$]{s1.east/0}{0}{x4/120,x1/-135,x2/-160};
\path (-1.2,0.8) node[above,lbl] {$a[S_1]$};
\hyperedgewithangles[0.15][$(3.6,-0.5)$]{x1/-80,x3/-35,x2/-70}{15}{s2.south/-100};
\path (3.6,-0.5) node[below,lbl] {$a'[S_2]$};
\hyperedgewithangles[0.3][$(3.2,0.75)$]{s2.west/180}{180}{x3/60,x2/-45,x1/-20};
\path (3.2,0.8) node[above,lbl] {$a[S_2]$};
\hyperedgewithangles[0.7][$(0,2.75)$]{x1/90,x2/170}{90}{s3.west/-170};
\path (0,2.5) node[right,lbl] {$a'[S_3]$};
\hyperedgewithangles[0.7][$(2,2.75)$]{s3.east/-10}{-90}{x1/10,x2/90};
\path (2,2.5) node[left,lbl] {$a[S_3]$};
\end{tikzpicture}
\end{center}
\caption{The directed hypergraph $\HH(\FF,D)$, with $D = \{x_1, \dots, x_4\}$ and $\FF$ consisting of $S_1 = \{x_1, x_2, x_4\}$, $S_2 = \{x_1, x_2, x_3\}$, and $S_3 = \{x_1,x_2\}$.}\label{fig:subset_hypergraph}
\end{figure}
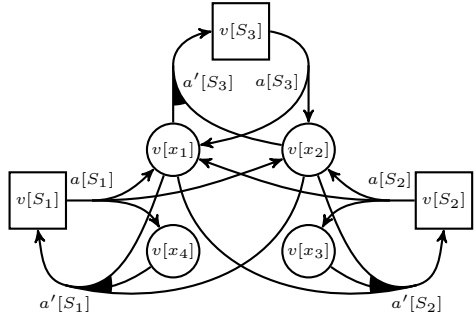

\begin{lemma}\label{lemma:reachable}
Given $S \in \FF$, $v$ is reachable from $v[S]$ in $\HH(\FF,D)$ if, and only if, $v = v[S']$ for some $S' \in \FF$ such that $S' \subseteq S$, or $v = v[x]$ for some $x \in S$.
\end{lemma}

\begin{proof}
Clearly, any vertex $v[x]$ is reachable from $v[S]$ through the hyperarc $a[S]$. Besides, assuming $S \supseteq S'$, then $v[S]$ reaches $v[S']$ through the hyperpath formed by the hyperarcs $a[S]$ and $a'[S']$.

Now, let us prove by induction that these are the only vertices reachable from $v[S]$. Let $u$ be reachable from $v[S]$. If $u = v[S]$, then this is obvious. Otherwise, there exists a hyperarc $a = (T,H)$ such that $u \in H$ and $T = \{u_1, \dots, u_q\}$ with each $u_i$ being reachable from $v[S]$. We can distinguish two cases:
\begin{enumerate}[(i)]
\item either $a$ is of the form $a[S']$ for some $S' \in \FF$, in which case the tail is reduced to the vertex $v[S']$, which is reachable from $v[S]$. By induction, we know that $S \supseteq S'$. Since $u = v[x]$ for some $x \in S'$, it follows that $x \in S$.
\item or $a$ is of the form $a'[S']$ for some $S' \in \FF$. Then its tail is the set of the $v[x]$ for $x \in S'$, and its head consists of the single vertex $v[S']$. Thus $x \in S$ for all $x \in S'$ by induction, which ensures that $u = v[S']$ with $S' \subseteq S$.\qedhere
\end{enumerate}
\end{proof}

\begin{proposition}\label{prop:transitive_reduction_lower_bound}
The size of the transitive reduction of the reachability relation of $\HH(\FF,D)$ is lower bounded by the size of the transitive reduction of the subset partial order over the family $\FF$. 
\end{proposition}

\begin{proof}
We claim that for any couple $(S,S')$ in the transitive reduction of the subset partial order over the family $\FF$, $(v[S'],v[S])$ belongs to the transitive reduction of the relation $\reach_{\HH(\FF,D)}$.

Suppose that the pair $(v[S'],v[S])$ is not in transitive reduction of $\reach_{\HH(\FF,D)}$, and that $S \subseteq S'$ (the case $S \not \subseteq S'$ is obvious). By Lemma~\ref{lemma:reachable}, $v[S]$ is reachable from $v[S']$. Besides, there exists a vertex $u$ of $\HH(\FF,D)$ distinct from $v[S]$ and $v[S']$ such that $v[S'] \reach_{\HH(\FF,D)} u \reach_{\HH(\FF,D)} v[S]$. Observe that any vertex reaching a vertex of the form $v[T]$ ($T \in \FF$) is necessarily of the form $v[T']$ for some $T' \in \FF$ (because of the assumption $\card{T} > 1$ which ensures that no vertex of the form $v[x]$ for $x \in D$ can reach $v[T]$). Consequently, there exists a set $S'' \in \FF$ (distinct from $S$ and $S'$) such that $u = v[S'']$. Following Lemma~\ref{lemma:reachable}, this shows that $S' \supsetneq S'' \supsetneq S$. Thus $(S,S')$ cannot belong to the transitive reduction of the subset partial order over $\FF$.
\end{proof}

The subset partial order has been well studied in the literature~\cite{YellinIPL93,PritchardIPL95,PritchardAlg99,PritchardJAlg99,ElmasryIPL09}. It has been proved in~\cite{YellinIPL93,ElmasryIPL09} that the size of the transitive reduction of the subset partial order can be superlinear in the size of the input $(\FF,D)$ (defined as $\card{D} + \sum_{S \in \FF} \card{S}$). Combining this with Proposition~\ref{prop:transitive_reduction_lower_bound} provides the following result:
\begin{theorem}\label{th:transitive_reduction_lower_bound}
There is a directed hypergraph $\HH$ such that the size of the transitive reduction of the reachability relation is in $\Omega(\size(\HH)^2 / \log^2 (\size(\HH)))$.
\end{theorem}

\begin{proof}
We use the construction given in~\cite{ElmasryIPL09} in which $\FF$ consists of two disjoint families $\FF_1$ and $\FF_2$ of sets over the domain $D = \{x_1, \dots, x_n\}$ (where $n$ is supposed to be divisible by 4). The first family consists of the subsets having $n/4$ elements among $x_1, \dots, x_{n/2}$. The second family is formed by the subsets containing all the elements $x_1, \dots, x_{n/2}$, and precisely $n/4$ elements among $x_{n/2+1}, \dots, x_n$. The transitive reduction of the subset partial order over $\FF$ coincides with the cartesian product $\FF_1 \times \FF_2$. Each $\FF_i$ precisely contains $\binom{n/2}{n/4} = \Theta(2^{n/2}/\sqrt{n})$ sets, so that the size of the transitive reduction of the subset partial order is $\Theta(2^n/n)$.

Proposition~\ref{prop:transitive_reduction_lower_bound} shows that the size of the transitive reduction of~$\reach_{\HH(\FF,D)}$ is in $\Omega(2^n/n)$. Now, the size of the directed hypergraph $\HH(\FF,D)$ is equal to:
\[
\size(\HH(\FF,D)) = n+2\binom{n/2}{n/4} + 2\frac{3n}{4}\binom{n/2}{n/4} + 2\frac{n}{4}\binom{n/2}{n/4},
\]
so that $\size(\HH(\FF,D)) = \Theta(\sqrt{n} 2^{n/2})$. This provides the expected result.
\end{proof}

The size of the transitive relation of the reachability relation can be seen as a partial measure of the complexity of the \scc{} computation problem. It is indeed natural to think at algorithms computing the \scc{}s by following the reachability relation between them, for instance by a depth-first search, hence by exploring at least a transitive reduction of the reachability relation. In fact, most of the algorithms determining \scc{}s of directed graphs, for instance the ones due to Tarjan~\cite{Tarjan72}, Cheriyan and Mehlhorn~\cite{CheriyanMehlhorn96}, or Gabow~\cite{Gabow00}, perform a depth-first search on the entire graph, and thus follow this approach. Theorem~\ref{th:transitive_reduction_lower_bound} shows that this class of algorithms cannot have a linear complexity on directed hypergraphs:
\begin{corollary}\label{cor:scc_lower_bound}
Any algorithm computing the strongly connected components of directed hypergraphs by traversing an entire transitive reduction of the reachability relation has a worst case complexity at least equal to $N^2 / \log^2 (N)$, where $N$ is the size of the input.
\end{corollary}
Consequently, the reachability relation must be sufficiently explored to identify the \scc{}s, but it cannot be totally explored unless sacrificing the time complexity. Note that the algorithm \Call{TerminalScc}{} relies on a certain trade-off to discover terminal \scc{}s: it only traverses hyperarcs $(T,H)$ such that $T$ is contained in a \scc{}, whereas hyperarcs in which the tail vertices belong to distinct \scc{}s are ignored.

\subsection{Reduction from the minimal set problem}\label{subsec:set_pb_reduction}

Given a family $\FF$ of distinct sets over a domain $D$ as above, the \emph{minimal set problem} consists in finding all minimal sets $S \in \FF$ for the subset partial order. This problem has received much attention~\cite{PritchardActInf91,YellinSODA92,YellinIPL93,PritchardIPL95,PritchardJAlg99,ElmasryIPL09,BayardoSDM11}. It has important applications in propositional logic~\cite{PritchardActInf91} or data mining~\cite{BayardoSDM11}. It can also be seen as a boolean case of the problem of finding maximal vectors among a given family~\cite{KungJACM75,KirkpatrickSOCG85,GodfreyVLDB05}. 

Surprisingly, the most efficient algorithms addressing the minimal set problem compute the whole subset partial order~\cite{YellinIPL93,ElmasryIPL09}. The best known methods to compute the subset partial order in the general case are due to Pritchard~\cite{PritchardIPL95,PritchardJAlg99}. Their complexity are in $O(N^2/\log N)$, where $N$ is the size of the input $(\FF,D)$. In the dense case, \ie\ when the size of the family is in $\Theta(\card{D} \cdot \card{\FF})$, Elmasry defined a method with a complexity in $O(N^2/\log^2 N)$~\cite{ElmasryIPL09}. This matches the lower bound provided in Corollary~\ref{cor:scc_lower_bound}.

In this section, we establish a linear time reduction from the minimal set problem to the problem of computing the \scc{}s in directed hypergraph. 
To obtain it, we build a directed hypergraph $\HHbar(\FF,D)$ starting from the hypergraph $\HH(\FF,D)$. On top of the vertices of the latter, $\HHbar(\FF,D)$ has the following vertices:
\begin{inparaenum}[(i)]
\item for each $S \in \FF$, an additional vertex $w[S]$,
\item $(\card{D}+1)$ vertices labelled by $c_0,\dots,c_{\card{D}}$,
\item and a special vertex labelled by $\supersetvtx$.
\end{inparaenum}
Besides, we add the following hyperarcs:
\begin{inparaenum}[(i)]
\item for each $S \in \FF$, a hyperarc leaving $\{v[S]\}$ and entering the singleton $\{c_{\card{S}-1}\}$,
\item for every $0 \leq i \leq \card{D}$, a hyperarc leaving $\{c_i\}$ and entering the set of the vertices $w[S]$ such that $\card{S} = i$,
\item for each $i > 0$, a hyperarc from $\{c_i\}$ to $\{c_{i-1}\}$,
\item for each $S \in \FF$, a hyperarc leaving the set $\{v[S],w[S]\}$ and entering the singleton $\{\supersetvtx\}$,
\item for every $S \in \FF$, a hyperarc from $\{\supersetvtx\}$ to $\{v[S]\}$.
\end{inparaenum}
This construction is illustrated in Figure~\ref{fig:minimal_subset_hypergraph}. 

\begin{figure}[t]
\begin{center}
\begin{tikzpicture}[scale=0.8,>=stealth',elt/.style={circle,inner sep=0.5pt,draw=black,thick,font=\tiny}, set/.style={rectangle, minimum height=0.6cm,minimum width=0.6cm,inner sep=1pt,draw=black,thick,font=\tiny}, hyperedge/.style={draw=black,thick}, simpleedge/.style={draw=black,thick},lbl/.style={font=\tiny}]
\node [elt] (x1) at (0,1.5) {$v[x_1]$};
\node [elt] (x2) at (2,1.5) {$v[x_2]$};
\node [elt] (x3) at (2,0)  {$v[x_3]$};
\node [elt] (x4) at (0,0)  {$v[x_4]$};

\node [set] (s1) at (-2,0.75) {$v[S_1]$};
\node [set] (s2) at (4,0.75) {$v[S_2]$};
\node [set] (s3) at (1,3.25) {$v[S_3]$};

\node [draw=black,rectangle,minimum height=0.7cm,dashed,very thick,font=\small] (superset) at (-2.75,3.5) {$\supersetvtx$};

\begin{scope}
\tikzset{hyperedge/.style={lightgray}};
\hyperedgewithangles[0.15][$(-1.6,-0.5)$]{x2/-100,x4/-145,x1/-110}{165}{s1.south/-80};
\hyperedgewithangles[0.3][$(-1.2,0.75)$]{s1.east/0}{0}{x4/120,x1/-135,x2/-160};
\hyperedgewithangles[0.15][$(3.6,-0.5)$]{x1/-80,x3/-35,x2/-70}{15}{s2.south/-100};
\hyperedgewithangles[0.3][$(3.2,0.75)$]{s2.west/180}{180}{x3/60,x2/-45,x1/-20};
\hyperedgewithangles[0.7][$(0,2.75)$]{x1/90,x2/170}{90}{s3.west/-170};
\hyperedgewithangles[0.7][$(2,2.75)$]{s3.east/-10}{-90}{x1/10,x2/90};
\end{scope}

\node [set] (w1) at (-3.5,0.75) {$w[S_1]$};
\node [set] (w2) at (5.5,0.75) {$w[S_2]$};
\node [set] (w3) at (2.5,3.25) {$w[S_3]$};

\node [elt] (c0) at (-2,-2) {$c_0$};
\node [elt] (c1) at (-0.5,-2) {$c_1$};
\node [elt] (c2) at (1,-2) {$c_2$};
\node [elt] (c3) at (2.5,-2) {$c_3$};
\node [elt] (c4) at (4,-2) {$c_4$};

\draw[simpleedge,->] (c4) -- (c3);
\draw[simpleedge,->] (c3) -- (c2);
\draw[simpleedge,->] (c2) -- (c1);
\draw[simpleedge,->] (c1) -- (c0);
\draw[simpleedge,->] (s1.240) to[out=-80,in=130] (-1.25,-2) to[out=-50,in=-120] (c2); %
\draw[simpleedge,->] (s2.-60) to[out=-100,in=50] (3.25,-2) to[out=-130,in=-60] (c2) ; 
\hyperedgewithangles[0.3][$(2.2,-1.7)$]{c3/135}{135}{w1/-90,w2/-90};
\draw[simpleedge,->] (s3) .. controls (1,-1) and (0,-1.5) .. (c1);
\draw[simpleedge,->] (c2) .. controls (1,2.25) and (2,2.75) .. (w3);
\hyperedgewithangles[0.6][$(-2.75,2.5)$]{w1/50,s1/130}{90}{superset/-90};
\hyperedgewithangles[0.6][$(4.75,2.5)$]{w2/130,s2/50}{90}{superset/60};
\hyperedgewithangles[0.5][$(-1.5,2.5)$]{s3/-135,w3/-110}{160}{superset/-45};
\draw[simpleedge,->] (superset.-70) to[out=-80,in=90] (s1);
\draw[simpleedge,->] (superset) to[out=0,in=180] (s3.160);
\draw[simpleedge,->] (superset.30) .. controls (-1.75,4) and (4,7).. (s2);

\end{tikzpicture}
\end{center}
\caption{The hypergraph $\HHbar(\FF,D)$, where $D = \{x_1, \dots, x_4\}$ and $\FF$ consists of $S_1 = \{x_1, x_2, x_4\}$, $S_2 = \{x_1, x_2, x_3\}$, and $S_3 = \{x_1,x_2\}$. The hyperarcs of $\HH(\FF,D)$ are depicted in gray.}\label{fig:minimal_subset_hypergraph}
\end{figure}

\begin{proposition}\label{prop:equivalence}
For any $S \in \FF$, $S$ is not minimal in $\FF$ if, and only if, the vertex $\supersetvtx$ is reachable from $v[S]$ in $\HHbar(\FF,D)$.
\end{proposition}

\begin{proof}
Assume that $S$ is not minimal in $\FF$, and let $S' \in \FF$ satisfying $S' \subsetneq S$. By Lemma~\ref{lemma:reachable}, $v[S']$ is reachable from $v[S]$ in $\HH(\FF,D)$, and hence in $\HHbar(\FF,D)$. Since $\card{S'} = j < \card{S} = i$, $w[S']$ is reachable from $v[S]$ through the hyperpath traversing the vertices $c_{i-1},c_{i-2},\dots,c_j$. Finally, the vertex $\supersetvtx$ is reachable through the hyperarc from $\{v[S'],w[S']\}$.

Conversely, suppose that $v[S]$ reaches $\supersetvtx$ in $\HHbar(\FF,D)$. Consider a minimal hyperpath $a_1, \dots, a_p$ from $v[S]$ to $\supersetvtx$. Necessarily, $a_p$ is a hyperarc of the form $(\{v[S'], w[S']\},\{\supersetvtx\})$ for some $S' \in \FF$. Consequently, both vertices $v[S']$ and $w[S']$ are reachable from $v[S]$. Besides, to each of the two vertices, there exists a hyperpath from $v[S]$, which is a subsequence of $a_1, \dots, a_{p-1}$, and which consequently does not contain the vertex $\supersetvtx$ (meaning that the latter does not appear in any tail or head of the hyperarcs of the hyperpath). 

Let $a'_1, \dots, a'_q$ be a minimal hyperpath from $v[S]$ to $v[S']$ not containing $\supersetvtx$. The hyperpath cannot contain any hyperarc of the form $(\{v[T], w[T]\},\{\supersetvtx\})$ (where $T \in \FF$). As a result, no vertex of the form $w[T]$ should occur in the hyperpath (by minimality). Similarly, no vertex of the form $c_i$ belongs to the hyperpath (otherwise, it should also contain a vertex of the form $w[T]$). It follows that the hyperpath $a'_1, \dots, a'_q$ is also a hyperpath in the hypergraph $\HH(\FF,D)$. Applying Lemma~\ref{lemma:reachable} then shows that $S' \subseteq S$. 

It remains to show that the latter inclusion is strict. Similarly, let $a''_1, \dots, a''_r$ be a minimal hyperpath from $v[S]$ to $w[S']$ not containing $\supersetvtx$. Then the tail of $a''_r$ is necessarily reduced to the vertex $c_i$, where $i = \card{S'}$, and its head is $\{w[S']\}$. It follows that $a''_1, \dots, a''_{r-1}$ is a hyperpath from $v[S]$ to $c_i$ not containing $\supersetvtx$. Now suppose that $i \geq \card{S}$. Let $j \geq i$ the greatest integer such that $c_j$ appears in the hyperpath $a''_1, \dots, a''_{r-1}$. Necessarily, one of the hyperarcs in the hyperpath is of the form $(\{v[T]\},\{c_j\})$, so that $v[T]$ is reachable from $v[S]$ through a hyperpath not passing through the vertex $\supersetvtx$. It follows from the previous discussion that $T \subseteq S$. But $\card{T} = j+1 > i \geq \card{S}$, which is a contradiction. This shows that $i = \card{S'} < \card{S}$, hence $S' \subsetneq S$.
\end{proof}

Since every vertex of the form $v[S]$ is reachable from $\supersetvtx$, minimal sets of the family $\FF$ are precisely given by the vertices which do not belong to the \scc{} of the vertex $\supersetvtx$. This proves the following complexity reduction:
\begin{theorem}\label{th:minimal_set_pb_reduction}
The minimal set problem can be reduced in linear time to the problem of determining the strongly connected components in a directed hypergraph.
\end{theorem}

\begin{proof}
We assume the existence of an oracle providing the \scc{}s of any directed hypergraph.	
	
Consider an instance $(\FF,D)$ of the minimal set problem. The hypergraph $\HHbar(\FF,D)$ can be built in linear time in the size of the input. Calling the oracle on $\HHbar(\FF,D)$ yields its \scc{}s. Then, by examining each \scc{} and its content, we collect the sets $S \in \FF$ such that $v[S]$ does not belong to the same component as the vertex $\supersetvtx$. We finally return these sets. By Proposition~\ref{prop:equivalence}, they are precisely the minimal sets in the family $\FF$.
\end{proof}

\begin{remark}
Another interesting combinatorial problem is to decide whether a collection of sets is a Sperner family, \ie\ the sets are not pairwise comparable. As a consequence of Theorem~\ref{th:minimal_set_pb_reduction}, it can be shown that the problem of deciding whether a collection of sets is a Sperner family can be reduced in linear time to the problem of determining the \scc{}s in a directed hypergraph. The Sperner family problem can be indeed reduced in linear time to the minimal set problem, by examining whether the number of minimal sets of $\FF$ is equal to the cardinality of $\FF$.
\end{remark}

\begin{remark}
In a similar way, we can also exhibit a linear time reduction from the problem of determining a linear extension of the subset partial order over a family of sets, to the problem of topologically sorting the vertices of an acyclic directed hypergraph. The \emph{topological sort} of an acyclic directed hypergraph $\HH$ refers to a total ordering $\preceq$ of the vertices such that $u \preceq v$ as soon as $u \reach_\HH v$. 

The idea is to use the directed hypergraph $\HH(\FF,D)$ (which can be built in linear time in the size of $(\FF,D)$). This hypergraph can be shown to be acyclic (under the assumption $\card{S} > 1$ for all $S \in \FF$). By Lemma~\ref{lemma:reachable}, it is straightforward that inverting and restricting a topological ordering over the vertices of the form $v[S]$ provides a linear extension of the partial order over~$\FF$.

To our knowledge, the problem of determining a linear extension of the subset partial order has not been particularly studied. It is probably not obvious to solve this problem without examining a significant part of the subset partial order (or at least of a sparse representation such as its transitive reduction). 
\end{remark}

\section{Conclusion} 
\label{subsec:complexity}

In this paper, we have proved that all terminal \scc{}s can be determined in only almost linear time (Theorems~\ref{th:correctness} and~\ref{th:complexity}). As a consequence, two other problems, testing strong connectivity and the existence of a sink, can be solved in almost linear time. 

The problem of computing all \scc{}s appears to be much harder. We conclude with the following questions:
\begin{open_problem}\label{op1}
Is it possible to compute the strongly connected components in directed hypergraphs with the same time and space complexity as in directed graphs? 
\end{open_problem}
\begin{open_problem}\label{op2}
Is it possible to ``break'' the partial lower bound $O(N^2/\log^2 N)$ provided by Corollary~\ref{cor:scc_lower_bound}?
\end{open_problem}
The results established in Section~\ref{sec:combinatorics} on the size of the transitive reduction of the reachability relation in hypergraphs (Theorem~\ref{th:transitive_reduction_lower_bound}), and on the reduction from the minimal set problem (Theorem~\ref{th:minimal_set_pb_reduction}), show that the answer to Question~\ref{op1} is likely to be ``No'' (at least considering ``reasonable'' models of computation, like the RAM model). Corollary~\ref{cor:scc_lower_bound} indicates that solving Question~\ref{op2} would require to design an algorithm capturing only a part of the reachability relation (or a transitive reduction). This part should be however sufficiently large to correctly identify the \scc{}s. In any case, the directed hypergraphs $\HH(\FF,D)$ and $\HHbar(\FF,D)$ constructed in Section~\ref{sec:combinatorics} provide useful examples to study the problem.

\bibliographystyle{amsalpha}

\begin{thebibliography}{LMM{\etalchar{+}}12}

\bibitem[ADS83]{AusielloJACM83}
Giorgio Ausiello, Alessandro D'Atri, and Domenico Sacc\`{a}, \emph{Graph
  algorithms for functional dependency manipulation}, J. ACM \textbf{30}
  (1983), 752--766.

\bibitem[ADS86]{Ausiello86}
G~Ausiello, A~D'Atri, and D~Sacc\'{a}, \emph{Minimal representation of directed
  hypergraphs}, SIAM J. Comput. \textbf{15} (1986), 418--431.

\bibitem[AFF01]{Ausiello01}
Giorgio Ausiello, Paolo~Giulio Franciosa, and Daniele Frigioni, \emph{Directed
  hypergraphs: Problems, algorithmic results, and a novel decremental
  approach}, Theoretical Computer Science, 7th Italian Conference, ICTCS 2001,
  Proceedings (Antonio Restivo, Simona Ronchi~Della Rocca, and Luca Roversi,
  eds.), Lecture Notes in Computer Science, vol. 2202, Springer, 2001,
  pp.~312--327.

\bibitem[AFFG97]{Ausiello97}
Giorgio Ausiello, Paolo~Giulio Franciosa, Daniele Frigioni, and Roberto
  Giaccio, \emph{Decremental maintenance of reachability in hypergraphs and
  minimum models of horn formulae}, Algorithms and Computation, 8th
  International Symposium, ISAAC '97, Singapore, December 17-19, 1997,
  Proceedings (Hon~Wai Leong, Hiroshi Imai, and Sanjay Jain, eds.), Lecture
  Notes in Computer Science, vol. 1350, Springer, 1997, pp.~122--131.

\bibitem[AGG08]{AllamigeonGaubertGoubaultSAS08}
X.~Allamigeon, S.~Gaubert, and E.~Goubault, \emph{Inferring min and max
  invariants using max-plus polyhedra}, Proceedings of the 15th International
  Static Analysis Symposium (SAS'08), Lecture Notes in Comput. Sci., vol. 5079,
  Springer, Valencia, Spain, 2008, pp.~189--204.

\bibitem[AGG10]{AllamigeonGaubertGoubaultSTACS10}
\bysame, \emph{The tropical double description method}, Proceedings of the 27th
  International Symposium on Theoretical Aspects of Computer Science (STACS
  2010) (Dagstuhl, Germany) (J.-Y. Marion and Th. Schwentick, eds.), Leibniz
  International Proceedings in Informatics (LIPIcs), vol.~5, Schloss
  Dagstuhl--Leibniz-Zentrum fuer Informatik, 2010, pp.~47--58.

\bibitem[AGG13]{AllamigeonGaubertGoubaultDCG2013}
\bysame, \emph{Computing the vertices of tropical polyhedra using directed
  hypergraphs}, Discrete \& Computational Geometry \textbf{49} (2013), no.~2,
  247--279.

\bibitem[AGU72]{AhoGareyUllmanSICOMP72}
Alfred~V. Aho, M.~R. Garey, and Jeffrey~D. Ullman, \emph{The transitive
  reduction of a directed graph}, SIAM Journal on Computing \textbf{1} (1972),
  no.~2, 131--137.

\bibitem[AI91]{Ausiello91}
Giorgio Ausiello and Giuseppe~F. Italiano, \emph{On-line algorithms for
  polynomially solvable satisfiability problems}, J. Log. Program. \textbf{10}
  (1991), no.~1/2/3{\&}4, 69--90.

\bibitem[AIL{\etalchar{+}}12]{AusielloISCO12}
Giorgio Ausiello, Giuseppe Italiano, Luigi Laura, Umberto Nanni, and Fabiano
  Sarracco, \emph{Structure theorems for optimum hyperpaths in directed
  hypergraphs}, Combinatorial Optimization (A.~Mahjoub, Vangelis Markakis,
  Ioannis Milis, and Vangelis Paschos, eds.), Lecture Notes in Computer
  Science, vol. 7422, Springer Berlin / Heidelberg, 2012, pp.~1--14.

\bibitem[All09]{tplib}
Xavier Allamigeon, \emph{{TPLib}: Tropical polyhedra library}, 2009,
  Distributed under LGPL, available at
  \url{https://gforge.inria.fr/projects/tplib}.

\bibitem[ANI90]{AusielloTCS90}
Giorgio Ausiello, Umberto Nanni, and Giuseppe~F. Italiano, \emph{Dynamic
  maintenance of directed hypergraphs}, Theoretical Computer Science
  \textbf{72} (1990), no.~2-3, 97 -- 117.

\bibitem[BP11]{BayardoSDM11}
Roberto~J. Bayardo and Biswanath Panda, \emph{Fast algorithms for finding
  extremal sets}, Proceedings of the Eleventh SIAM International Conference on
  Data Mining, SDM 2011, April 28-30, 2011, Mesa, Arizona, USA, SIAM /
  Omnipress, 2011, pp.~25--34.

\bibitem[CM96]{CheriyanMehlhorn96}
J.~Cheriyan and K.~Mehlhorn, \emph{Algorithms for dense graphs and networks on
  the random access computer}, Algorithmica \textbf{15} (1996), 521--549.

\bibitem[CSRL01]{Cormen01}
Thomas~H. Cormen, Clifford Stein, Ronald~L. Rivest, and Charles~E. Leiserson,
  \emph{Introduction to algorithms}, McGraw-Hill Higher Education, 2001.

\bibitem[Elm09]{ElmasryIPL09}
Amr Elmasry, \emph{Computing the subset partial order for dense families of
  sets}, Information Processing Letters \textbf{109} (2009), no.~18, 1082 --
  1086.

\bibitem[Gab00]{Gabow00}
Harold~N. Gabow, \emph{Path-based depth-first search for strong and biconnected
  components}, Inf. Process. Lett. \textbf{74} (2000), no.~3-4, 107--114.

\bibitem[GG04]{GaubertGunawardena04}
S.~Gaubert and J.~Gunawardena, \emph{The {P}erron-{F}robenius theorem for
  homogeneous, monotone functions}, Trans. of AMS \textbf{356} (2004), no.~12,
  4931--4950.

\bibitem[GGPR98]{Gallo98}
Giorgio Gallo, Claudio Gentile, Daniele Pretolani, and Gabriella Rago,
  \emph{Max horn sat and the minimum cut problem in directed hypergraphs},
  Math. Program. \textbf{80} (1998), 213--237.

\bibitem[GLPN93]{GalloDAM93}
Giorgio Gallo, Giustino Longo, Stefano Pallottino, and Sang Nguyen,
  \emph{Directed hypergraphs and applications}, Discrete Appl. Math.
  \textbf{42} (1993), no.~2-3, 177--201.

\bibitem[GP95]{Gallo95}
Giorgio Gallo and Daniele Pretolani, \emph{A new algorithm for the
  propositional satisfiability problem}, Discrete Applied Mathematics
  \textbf{60} (1995), no.~1-3, 159--179.

\bibitem[GSG05]{GodfreyVLDB05}
Parke Godfrey, Ryan Shipley, and Jarek Gryz, \emph{Maximal vector computation
  in large data sets}, Proceedings of the 31st international conference on Very
  large data bases, VLDB '05, VLDB Endowment, 2005, pp.~229--240.

\bibitem[Kat07]{katz05}
R.~D. Katz, \emph{Max-plus {$(A,B)$}-invariant spaces and control of timed
  discrete event systems}, IEEE Trans. Aut. Control \textbf{52} (2007), no.~2,
  229--241.

\bibitem[KLP75]{KungJACM75}
H.~T. Kung, F.~Luccio, and F.~P. Preparata, \emph{On finding the maxima of a
  set of vectors}, J. ACM \textbf{22} (1975), 469--476.

\bibitem[KS85]{KirkpatrickSOCG85}
David~G. Kirkpatrick and Raimund Seidel, \emph{Output-size sensitive algorithms
  for finding maximal vectors}, Proceedings of the first annual symposium on
  Computational geometry (New York, NY, USA), SCG '85, ACM, 1985, pp.~89--96.

\bibitem[LMM{\etalchar{+}}12]{LuJLAP2011}
Qi~Lu, Michael Madsen, Martin Milata, S{\o}ren Ravn, Uli Fahrenberg, and Kim~G.
  Larsen, \emph{Reachability analysis for timed automata using max-plus
  algebra}, The Journal of Logic and Algebraic Programming \textbf{81} (2012),
  no.~3, 298--313.

\bibitem[LS98]{Liu98}
Xinxin Liu and Scott~A. Smolka, \emph{Simple linear-time algorithms for minimal
  fixed points (extended abstract)}, Automata, Languages and Programming, 25th
  International Colloquium, ICALP'98, Proceedings (Kim~Guldstrand Larsen, Sven
  Skyum, and Glynn Winskel, eds.), Lecture Notes in Computer Science, vol.
  1443, Springer, 1998, pp.~53--66.

\bibitem[NP89]{Nguyen89}
S.~Nguyen and S.~Pallottino, \emph{Hyperpaths and shortest hyperpaths}, COMO
  '86: Lectures given at the third session of the Centro Internazionale
  Matematico Estivo (C.I.M.E.) on Combinatorial optimization (New York, NY,
  USA), Lectures Notes in Mathematics, Springer-Verlag New York, Inc., 1989,
  pp.~258--271.

\bibitem[NPA06]{NielsenORL06}
Lars~Relund Nielsen, Daniele Pretolani, and Kim~Allan Andersen, \emph{Finding
  the k shortest hyperpaths using reoptimization}, Operations Research Letters
  \textbf{34} (2006), no.~2, 155 -- 164.

\bibitem[NPG98]{Nguyen98}
Sang Nguyen, Stefano Pallottino, and Michel Gendreau, \emph{Implicit
  enumeration of hyperpaths in a logit model for transit networks},
  Transportation Science \textbf{32} (1998), no.~1, 54--64.

\bibitem[{\"O}zt08]{Ozturan08}
Can~C. {\"O}zturan, \emph{On finding hypercycles in chemical reaction
  networks}, Appl. Math. Lett. \textbf{21} (2008), no.~9, 881--884.

\bibitem[Pre00]{Pretolani00}
Daniele Pretolani, \emph{A directed hypergraph model for random time dependent
  shortest paths}, European Journal of Operational Research \textbf{123}
  (2000), no.~2, 315--324.

\bibitem[Pre03]{Pretolani03}
Daniele Pretolani, \emph{Hypergraph reductions and satisfiability problems},
  Theory and Applications of Satisfiability Testing, 6th International
  Conference, SAT 2003 (Enrico Giunchiglia and Armando Tacchella, eds.),
  Lecture Notes in Computer Science, vol. 2919, Springer, 2003, pp.~383--397.

\bibitem[Pri91]{PritchardActInf91}
Paul Pritchard, \emph{Opportunistic algorithms for eliminating supersets}, Acta
  Informatica \textbf{28} (1991), 733--754.

\bibitem[Pri95]{PritchardIPL95}
Paul Pritchard, \emph{A simple sub-quadratic algorithm for computing the subset
  partial order}, Information Processing Letters \textbf{56} (1995), no.~6, 337
  -- 341.

\bibitem[Pri99a]{PritchardAlg99}
Paul Pritchard, \emph{A fast bit-parallel algorithm for computing the subset
  partial order}, Algorithmica \textbf{24} (1999), 76--86.

\bibitem[Pri99b]{PritchardJAlg99}
Paul Pritchard, \emph{On computing the subset graph of a collection of sets},
  Journal of Algorithms \textbf{33} (1999), no.~2, 187 -- 203.

\bibitem[Tar72]{Tarjan72}
Robert Tarjan, \emph{Depth-first search and linear graph algorithms}, SIAM
  Journal on Computing \textbf{1} (1972), no.~2, 146--160.

\bibitem[TT09]{ThakurTripathiTCS09}
Mayur Thakur and Rahul Tripathi, \emph{Linear connectivity problems in directed
  hypergraphs}, Theor. Comput. Sci. \textbf{410} (2009), 2592--2618.

\bibitem[Yel92]{YellinSODA92}
Daniel~M. Yellin, \emph{Algorithms for subset testing and finding maximal
  sets}, Proceedings of the third annual ACM-SIAM symposium on Discrete
  algorithms (Philadelphia, PA, USA), SODA '92, Society for Industrial and
  Applied Mathematics, 1992, pp.~386--392.

\bibitem[YJ93]{YellinIPL93}
Daniel~M. Yellin and Charanjit~S. Jutla, \emph{Finding extremal sets in less
  than quadratic time}, Information Processing Letters \textbf{48} (1993),
  no.~1, 29 -- 34.

\end{thebibliography}
\newcommand{\etalchar}[1]{$^{#1}$}
\providecommand{\bysame}{\leavevmode\hbox to3em{\hrulefill}\thinspace}
\providecommand{\MR}{\relax\ifhmode\unskip\space\fi MR }
\providecommand{\MRhref}[2]{%
  \href{http://www.ams.org/mathscinet-getitem?mr=#1}{#2}
}
\providecommand{\href}[2]{#2}

\appendix

\section{An Example of Complete Execution Trace of the Algorithm of Section~\ref{sec:maxscc}}\label{sec:execution_trace}

We give the main steps of the execution of the algorithm~\Call{TerminalScc}{} on the directed hypergraph depicted in Figure~\ref{fig:hypergraph}:
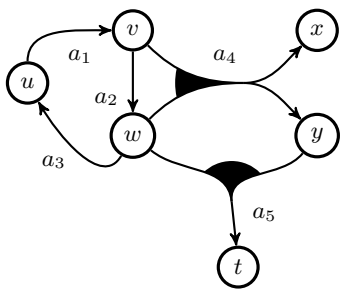
\begin{center}
\begin{small}
\begin{tikzpicture}[>=stealth',scale=0.7, vertex/.style={circle,draw=black,very thick,minimum size=2ex}, hyperedge/.style={draw=black,thick},simpleedge/.style={thick}]
\node [vertex] (u) at (-2,-1) {$u$};
\node [vertex] (v) at (0,0) {$v$};
\node [vertex] (w) at (0,-2)  {$w$};
\node [vertex] (x) at (3.5,0) {$x$};
\node [vertex] (y) at (3.5,-2) {$y$};
\node [vertex] (t) at (2,-4.5) {$t$};

\path[->] (u) edge[simpleedge,out=90,in=-180] (v);
\node at (-1,-0.5) {$a_1$};
\path[->] (v) edge[simpleedge,out=-90,in=90] (w);
\node at (-0.5,-1.3) {$a_2$};
\path[->] (w) edge[simpleedge,out=-120,in=-60] (u);
\node at (-1.5,-2.5) {$a_3$};
\node at (1.75,-0.5) {$a_4$};
\node at (2.5,-3.5) {$a_5$};
\oldhyperedge{v}{v,w}{w}{0.35}{0.6}{y,x};
\oldhyperedge{y}{y,w}{w}{0.48}{0.5}{t};
\end{tikzpicture}
\end{small}
\end{center}
Vertices are depicted by solid circles if their index is defined, and
by dashed circles otherwise. Once a vertex is placed into $\incomp$, it is depicted in gray. Similarly, a hyperarc which has never been placed into a local stack $F$ is represented by dotted lines. Once it is pushed into $F$, it becomes solid, and when it is popped from $F$, it is colored in gray (note that for the sake of readability, gray hyperarcs mapped to trivial cycles after a vertex merging step will not be represented). The stack $F$ which is mentioned always corresponds to the stack local to the last non-terminated call of the function $\Call{Visit}{}$.

Initially, $\Find{z} = z$ for all $z \in \{ u,v,w,x,y,t \}$. We suppose that $\Call{Visit}{u}$ is called first. After the execution of the block from Lines~\lineref{scc:begin} to~\lineref{scc:end_node_loop}, the current state is:
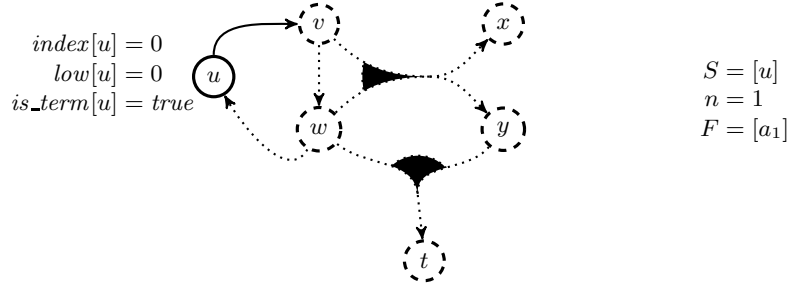
\begin{center}
\begin{small}
\begin{tikzpicture}[>=stealth',scale=0.7, vertex/.style={circle,draw=black,very thick,minimum size=2ex}, hyperedge/.style={draw=black,thick,dotted}, simpleedge/.style={draw=black,thick}]
\node [vertex] (u) at (-2,-1) {$u$} node[node distance=11ex,left of=u] {
$\begin{aligned} 
\tindex[u] &= 0 \\[-0.8ex]
\troot[u] &= 0 \\[-0.8ex]
\ismax[u] &= \True 
\end{aligned}$
};
\node [vertex,dashed] (v) at (0,0) {$v$};
\node [vertex,dashed] (w) at (0,-2)  {$w$};
\node [vertex,dashed] (x) at (3.5,0) {$x$};
\node [vertex,dashed] (y) at (3.5,-2) {$y$};
\node [vertex,dashed] (t) at (2,-4.5) {$t$};

\path[->] (u) edge[simpleedge,out=90,in=-180] (v);
\path[->] (v) edge[simpleedge,out=-90,in=90,dotted] (w);
\path[->] (w) edge[simpleedge,out=-120,in=-60,dotted] (u);
\oldhyperedge{v}{v,w}{w}{0.35}{0.6}{y,x};
\oldhyperedge{y}{y,w}{w}{0.48}{0.5}{t};

\node[anchor=west] at (7, -1.5) {$\begin{aligned} 
S & = [u] \\[-0.8ex]
n &= 1 \\[-0.8ex]
F & = [a_1]
\end{aligned}$ };
\end{tikzpicture}
\end{small}
\end{center}
Following the hyperarc $a_1$, $\Call{Visit}{v}$ is called during the execution of the block from Lines~\lineref{scc:begin_edge_loop} to~\lineref{scc:end_edge_loop} of $\Call{Visit}{u}$. After Line~\lineref{scc:end_node_loop} in $\Call{Visit}{v}$, the root of the hyperarc $a_4$ is set to $v$, and the counter $c_{a_4}$ is incremented to $1$ since $v \in S$. The state is:
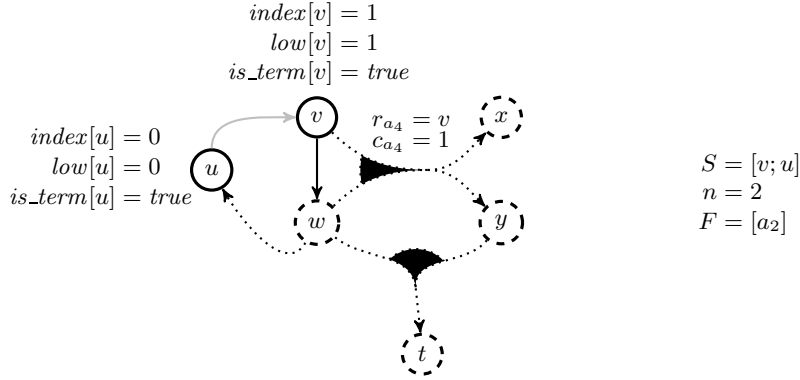
\begin{center}
\begin{small}
\begin{tikzpicture}[>=stealth',scale=0.7, vertex/.style={circle,draw=black,very thick,minimum size=2ex}, hyperedge/.style={draw=black,thick,dotted}, simpleedge/.style={draw=black,thick}]
\node [vertex] (u) at (-2,-1) {$u$} node[node distance=11ex,left of=u] {
$\begin{aligned} 
\tindex[u] &= 0 \\[-0.8ex]
\troot[u] &= 0 \\[-0.8ex]
\ismax[u] &= \True 
\end{aligned}$
};
\node [vertex] (v) at (0,0) {$v$} node[node distance=7ex,above of=v] {
$\begin{aligned} 
\tindex[v] &= 1 \\[-0.8ex]
\troot[v] &= 1 \\[-0.8ex]
\ismax[v] &= \True 
\end{aligned}$
};
\node [vertex,dashed] (w) at (0,-2)  {$w$};
\node [vertex,dashed] (x) at (3.5,0) {$x$};
\node [vertex,dashed] (y) at (3.5,-2) {$y$};
\node [vertex,dashed] (t) at (2,-4.5) {$t$};

\path[->] (u) edge[simpleedge,out=90,in=-180,lightgray] (v);
\path[->] (v) edge[simpleedge,out=-90,in=90] (w);
\path[->] (w) edge[simpleedge,out=-120,in=-60,dotted] (u);
\oldhyperedge{v}{v,w}{w}{0.35}{0.6}{y,x};
\oldhyperedge{y}{y,w}{w}{0.48}{0.5}{t};

\node at (1.75,-0.3) {$\begin{aligned} r_{a_4} & = v \\[-1.5ex] c_{a_4} & = 1 \end{aligned}$};

\node[anchor=west] at (7, -1.5) {$\begin{aligned} 
S & = [v; u] \\[-0.8ex]
n &= 2 \\[-0.8ex]
F & = [a_2]
\end{aligned}$ };
\end{tikzpicture}
\end{small}
\end{center}
Similarly, the function $\Call{Visit}{w}$ is called during the execution of the loop from Lines~\lineref{scc:begin_edge_loop} to~\lineref{scc:end_edge_loop} in $\Call{Visit}{v}$. After Line~\lineref{scc:end_node_loop} in $\Call{Visit}{w}$, the root of the hyperarc $a_5$ is set to $w$, and the counter $c_{a_5}$ is incremented to $1$ since $w \in S$. Besides, $c_{a_4}$ is incremented to $2 = \card{T(a_4)}$ since $\Find{r_{a_4}} = \Find{v} = v \in S$, so that $a_4$ is pushed on the stack $F_v$. The state is:
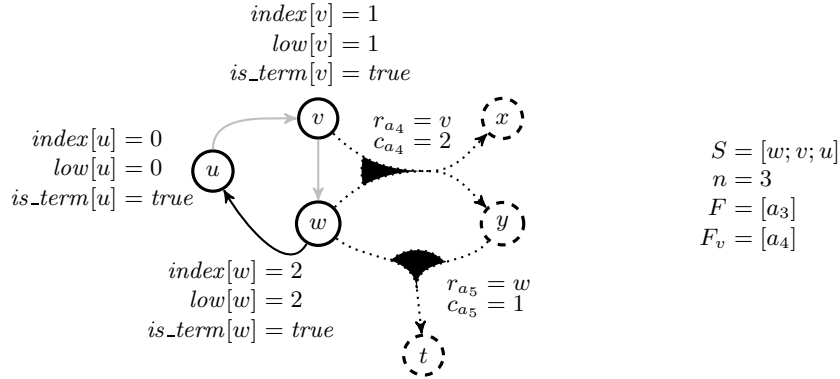
\begin{center}
\begin{small}
\begin{tikzpicture}[>=stealth',scale=0.7, vertex/.style={circle,draw=black,very thick,minimum size=2ex}, hyperedge/.style={draw=black,thick,dotted}, simpleedge/.style={draw=black,thick}]
\node [vertex] (u) at (-2,-1) {$u$} node[node distance=11ex,left of=u] {
$\begin{aligned} 
\tindex[u] &= 0 \\[-0.8ex]
\troot[u] &= 0 \\[-0.8ex]
\ismax[u] &= \True 
\end{aligned}$
};
\node [vertex] (v) at (0,0) {$v$} node[node distance=7ex,above of=v] {
$\begin{aligned} 
\tindex[v] &= 1 \\[-0.8ex]
\troot[v] &= 1 \\[-0.8ex]
\ismax[v] &= \True 
\end{aligned}$
};
\node [vertex] (w) at (0,-2) {$w$} node[node distance=11ex,below left of=w] {
$\begin{aligned} 
\tindex[w] &= 2 \\[-0.8ex]
\troot[w] &= 2 \\[-0.8ex]
\ismax[w] &= \True 
\end{aligned}$
};
\node [vertex,dashed] (x) at (3.5,-0) {$x$};
\node [vertex,dashed] (y) at (3.5,-2) {$y$};
\node [vertex,dashed] (t) at (2,-4.5) {$t$};

\node at (1.75,-0.3) {$\begin{aligned} r_{a_4} & = v \\[-1.5ex] c_{a_4} & = 2 \end{aligned}$};
\node at (3.2,-3.4) {$\begin{aligned} r_{a_5} & = w \\[-1.5ex] c_{a_5} & = 1 \end{aligned}$};

\path[->] (u) edge[simpleedge,out=90,in=-180,lightgray] (v);
\path[->] (v) edge[simpleedge,out=-90,in=90,lightgray] (w);
\path[->] (w) edge[simpleedge,out=-120,in=-60] (u);
\oldhyperedge{v}{v,w}{w}{0.35}{0.6}{y,x};
\oldhyperedge{y}{y,w}{w}{0.48}{0.5}{t};

\node[anchor=west] at (7, -1.5) {$\begin{aligned} 
S & = [w; v; u] \\[-0.8ex]
n &= 3 \\[-0.8ex] 
F & = [a_3]\\[-0.8ex]
F_{v} & = [a_4] 
\end{aligned}$ };
\end{tikzpicture}
\end{small}
\end{center}
The execution of the loop from Lines~\lineref{scc:begin_edge_loop} to~\lineref{scc:end_edge_loop} of $\Call{Visit}{w}$ discovers that $\tindex[u]$ is defined but $u \not \in \incomp$, so that $\troot[w]$ is set to $\min(\troot[w],\troot[u]) = 0$ and $\ismax[w]$ to $\ismax[w] \And \ismax[u] = \True$. At the end of the loop, the state is therefore:
\begin{center}
\begin{small}
\begin{tikzpicture}[>=stealth',scale=0.7, vertex/.style={circle,draw=black,very thick,minimum size=2ex}, hyperedge/.style={draw=black,thick,dotted}, simpleedge/.style={draw=black,thick}]
\node [vertex] (u) at (-2,-1) {$u$} node[node distance=11ex,left of=u] {
$\begin{aligned} 
\tindex[u] &= 0 \\[-0.8ex]
\troot[u] &= 0 \\[-0.8ex]
\ismax[u] &= \True 
\end{aligned}$
};
\node [vertex] (v) at (0,0) {$v$} node[node distance=7ex,above of=v] {
$\begin{aligned} 
\tindex[v] &= 1 \\[-0.8ex]
\troot[v] &= 1 \\[-0.8ex]
\ismax[v] &= \True 
\end{aligned}$
};
\node [vertex] (w) at (0,-2)  {$w$} node[node distance=11ex,below left of=w] {
$\begin{aligned} 
\tindex[w] &= 2 \\[-0.8ex]
\troot[w] &= 0 \\[-0.8ex]
\ismax[w] &= \True 
\end{aligned}$
};
\node [vertex,dashed] (x) at (3.5,0) {$x$};
\node [vertex,dashed] (y) at (3.5,-2) {$y$};
\node [vertex,dashed] (t) at (2,-4.5) {$t$};

\path[->] (u) edge[simpleedge,out=90,in=-180,lightgray] (v);
\path[->] (v) edge[simpleedge,out=-90,in=90,lightgray] (w);
\path[->] (w) edge[simpleedge,out=-120,in=-60,lightgray] (u);
\oldhyperedge{v}{v,w}{w}{0.35}{0.6}{y,x};
\oldhyperedge{y}{y,w}{w}{0.48}{0.5}{t};

\node at (1.75,-0.3) {$\begin{aligned} r_{a_4} & = v \\[-1.5ex] c_{a_4} & = 2 \end{aligned}$};
\node at (3.2,-3.4) {$\begin{aligned} r_{a_5} & = w \\[-1.5ex] c_{a_5} & = 1 \end{aligned}$};

\node[anchor=west] at (7, -1.5) {$\begin{aligned} 
S & = [w; v; u] \\[-0.8ex]
n &= 3 \\[-0.8ex]
F & = \emptystack\\[-0.8ex]
F_{v} & = [a_4] 
\end{aligned}$ };
\end{tikzpicture}
\end{small}
\end{center}
Since $\troot[w] \neq \tindex[w]$, the block from Lines~\lineref{scc:begin2} to~\lineref{scc:end2} is not executed, and $\Call{Visit}{w}$ terminates. Back to the loop from Lines~\lineref{scc:begin_edge_loop} to~\lineref{scc:end_edge_loop} in $\Call{Visit}{v}$, $\troot[v]$ is assigned to the value $\min(\troot[v],\troot[w]) = 0$, and $\ismax[v]$ to $\ismax[v] \And \ismax[w] = \True$:
\begin{center}
\begin{small}
\begin{tikzpicture}[>=stealth',scale=0.7, vertex/.style={circle,draw=black,very thick,minimum size=2ex}, hyperedge/.style={draw=black,thick,dotted}, simpleedge/.style={draw=black,thick}]
\node [vertex] (u) at (-2,-1) {$u$} node[node distance=11ex,left of=u] {
$\begin{aligned} 
\tindex[u] &= 0 \\[-0.8ex]
\troot[u] &= 0 \\[-0.8ex]
\ismax[u] &= \True 
\end{aligned}$
};
\node [vertex] (v) at (0,-0) {$v$} node[node distance=7ex,above of=v] {
$\begin{aligned} 
\tindex[v] &= 1 \\[-0.8ex]
\troot[v] &= 0 \\[-0.8ex]
\ismax[v] &= \True 
\end{aligned}$
};
\node [vertex] (w) at (0,-2)  {$w$} node[node distance=11ex,below left of=w] {
$\begin{aligned} 
\tindex[w] &= 2 \\[-0.8ex]
\troot[w] &= 0 \\[-0.8ex]
\ismax[w] &= \True 
\end{aligned}$
};

\node [vertex,dashed] (x) at (3.5,0) {$x$};
\node [vertex,dashed] (y) at (3.5,-2) {$y$};
\node [vertex,dashed] (t) at (2,-4.5) {$t$};

\path[->] (u) edge[simpleedge,out=90,in=-180,lightgray] (v);
\path[->] (v) edge[simpleedge,out=-90,in=90,lightgray] (w);
\path[->] (w) edge[simpleedge,out=-120,in=-60,lightgray] (u);
\oldhyperedge{v}{v,w}{w}{0.35}{0.6}{y,x};
\oldhyperedge{y}{y,w}{w}{0.48}{0.5}{t};

\node at (1.75,-0.3) {$\begin{aligned} r_{a_4} & = v \\[-1.5ex] c_{a_4} & = 2 \end{aligned}$};
\node at (3.2,-3.4) {$\begin{aligned} r_{a_5} & = w \\[-1.5ex] c_{a_5} & = 1 \end{aligned}$};

\node[anchor=west] at (7, -1.5) {$\begin{aligned} 
S & = [w; v; u] \\[-0.8ex]
n &= 3 \\[-0.8ex]
F & = \emptystack\\[-0.8ex]
F_{v} & = [a_4] 
\end{aligned}$ };
\end{tikzpicture}
\end{small}
\end{center}
Since $\troot[v] \neq \tindex[v]$, the block from Lines~\lineref{scc:begin2} to~\lineref{scc:end2} is not executed, and $\Call{Visit}{v}$ terminates. Back to the loop from Lines~\lineref{scc:begin_edge_loop} to~\lineref{scc:end_edge_loop} in $\Call{Visit}{u}$, $\troot[u]$ is assigned to the value $\min(\troot[u],\troot[v]) = 0$, and $\ismax[u]$ to $\ismax[u] \And \ismax[v] = \True$. Therefore, at Line~\lineref{scc:begin2}, the conditions $\troot[u] = \tindex[u]$ and $\ismax[u] = \True$ hold, so that a vertex merging step is executed. At that point, the stack $F$ is empty. After that, $i$ is set to $\tindex[u] = 0$ (Line~\lineref{scc:begin_node_merging}), and $F_u = \emptystack$ is emptied to $F$ (Line~\lineref{scc:push_on_fprime1}), so that $F$ is still empty. Then $w$ is popped from $S$, and since $\tindex[w] = 2 > i = 0$, the loop from Lines~\lineref{scc:begin_node_merging_loop} to~\lineref{scc:end_node_merging_loop} is iterated. Then the stack $F_w = \emptystack$ is emptied in $F$. At Line~\lineref{scc:merge}, \Call{Merge}{$u,w$} is called. The result is denoted by $U$ (in practice, either $U = u$ or $U = w$). The state is:
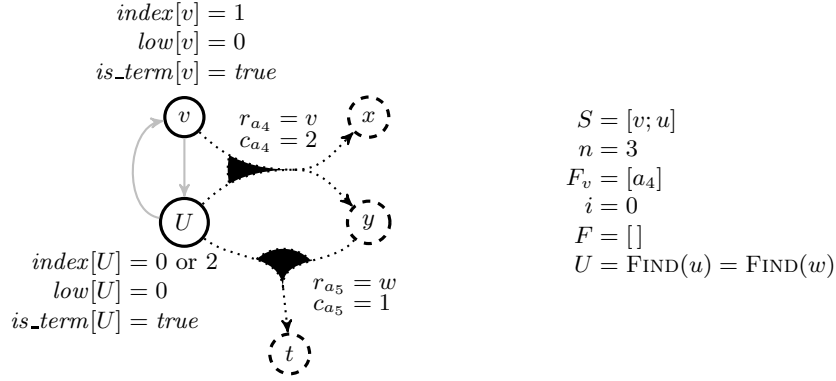
\begin{center}
\begin{small}
\begin{tikzpicture}[>=stealth',scale=0.7, vertex/.style={circle,draw=black,very thick,minimum size=2ex}, hyperedge/.style={draw=black,thick,dotted}, simpleedge/.style={draw=black,thick}]
\node [vertex] (v) at (0,0) {$v$} node[node distance=7ex,above of=v] {
$\begin{aligned} 
\tindex[v] &= 1 \\[-0.8ex]
\troot[v] &= 0 \\[-0.8ex]
\ismax[v] &= \True 
\end{aligned}$
};
\node [vertex] (w) at (0,-2)  {$U$} node[node distance=10ex,below left of=w] {
$\begin{aligned} 
\tindex[U] &= 0 \text{ or } 2  \\[-0.8ex]
\troot[U] &= 0 \\[-0.8ex]
\ismax[U] &= \True 
\end{aligned}$
};
\node [vertex,dashed] (x) at (3.5,0) {$x$};
\node [vertex,dashed] (y) at (3.5,-2) {$y$};
\node [vertex,dashed] (t) at (2,-4.5) {$t$};

\path[->] (v) edge[simpleedge,out=-90,in=90,lightgray] (w);
\path[->] (w) edge[simpleedge,out=170,in=-170,lightgray] (v);
\oldhyperedge{v}{v,w}{w}{0.35}{0.6}{y,x};
\oldhyperedge{y}{y,w}{w}{0.48}{0.5}{t};

\node at (1.75,-0.3) {$\begin{aligned} r_{a_4} & = v \\[-1.5ex] c_{a_4} & = 2 \end{aligned}$};
\node at (3.2,-3.4) {$\begin{aligned} r_{a_5} & = w \\[-1.5ex] c_{a_5} & = 1 \end{aligned}$};

\node[anchor=west] at (7, -1.5) {$\begin{aligned} 
S & = [v; u] \\[-0.8ex]
n &= 3 \\[-0.8ex]
F_{v} & = [a_4] \\[-0.8ex]
i & = 0 \\[-0.8ex]
F & = \emptystack\\[-0.8ex]
U & = \Find{u} = \Find{w} 
\end{aligned}$ };
\end{tikzpicture}
\end{small}
\end{center}
Then $v$ is popped from $S$, and since $\tindex[v] = 1 > i = 0$, the loop Lines~\lineref{scc:begin_node_merging_loop} to~\lineref{scc:end_node_merging_loop} is iterated again. Then the stack $F_v = [a_4]$ is emptied in $F$. At Line~\lineref{scc:merge}, \Call{Merge}{$U,v$} is called. The result is set to $U$ (in practice, $U$ is one of the vertices $u$, $v$, $w$). The state is:
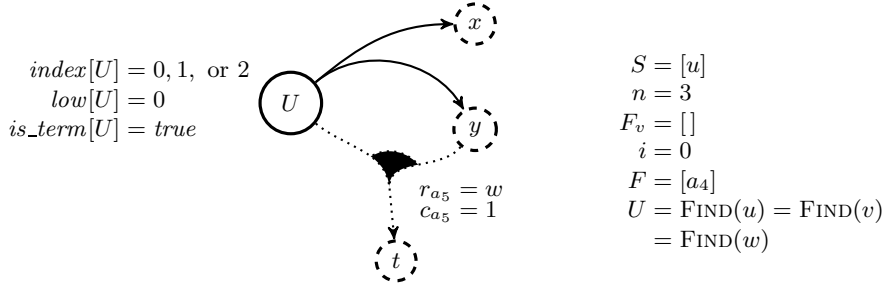
\begin{center}
\begin{small}
\begin{tikzpicture}[>=stealth',scale=0.7, vertex/.style={circle,draw=black,very thick,minimum size=2ex}, hyperedge/.style={draw=black,thick,dotted}, simpleedge/.style={draw=black,thick}]
\node [vertex,very thick,text width=0.48cm,text centered] (U) at (0,-1.5) {$U$} node[node distance=16ex,left of=U] {
$\begin{aligned} 
\tindex[U] &= 0, 1, \text{ or } 2  \\[-0.8ex]
\troot[U] &= 0 \\[-0.8ex]
\ismax[U] &= \True 
\end{aligned}$
};
\node [vertex,dashed] (x) at (3.5,0) {$x$};
\node [vertex,dashed] (y) at (3.5,-2) {$y$};
\node [vertex,dashed] (t) at (2,-4.5) {$t$};

\node at (3.2,-3.4) {$\begin{aligned} r_{a_5} & = w \\[-1.5ex] c_{a_5} & = 1 \end{aligned}$};
\path[->] (U) edge[simpleedge,out=40,in=-180] (x);
\path[->] (U) edge[simpleedge,out=40,in=120] (y);
\oldhyperedge{y}{y,U}{U}{0.48}{0.5}{t};

\node[anchor=west] at (6, -2.5) {$\begin{aligned} 
S & = [u] \\[-0.8ex]
n &= 3 \\[-0.8ex]
F_{v} & = \emptystack \\[-0.8ex]
i & = 0 \\[-0.8ex]
F & = [a_4]\\[-0.8ex]
U & = \Find{u} = \Find{v} \\[-0.8ex]
& = \Find{w} 
\end{aligned}$ };
\end{tikzpicture}
\end{small}
\end{center}
After that, $u$ is popped from $S$, and as $\tindex[u] = 0 = i$, the loop is terminated. At Line~\lineref{scc:index_redef}, $\tindex[U]$ is set to $i$, and $U$ is pushed on $S$. Since $F \neq \emptyset$, we go back to Line~\lineref{scc:begin_edge_loop}, in the state:
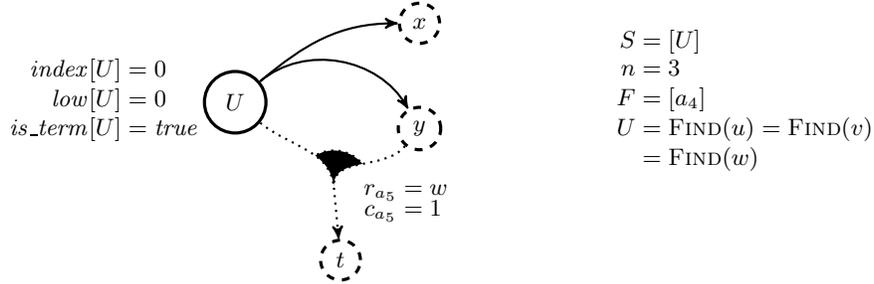
\begin{center}
\begin{small}
\begin{tikzpicture}[>=stealth',scale=0.7, vertex/.style={circle,draw=black,very thick,minimum size=2ex}, hyperedge/.style={draw=black,thick,dotted}, simpleedge/.style={draw=black,thick}]
\node [vertex,very thick,text width=0.48cm,text centered] (U) at (0,-1.5) {$U$} node[node distance=13ex,left of=U] {
$\begin{aligned} 
\tindex[U] &= 0  \\[-0.8ex]
\troot[U] &= 0 \\[-0.8ex]
\ismax[U] &= \True 
\end{aligned}$
};
\node [vertex,dashed] (x) at (3.5,0) {$x$};
\node [vertex,dashed] (y) at (3.5,-2) {$y$};
\node [vertex,dashed] (t) at (2,-4.5) {$t$};

\node at (3.2,-3.4) {$\begin{aligned} r_{a_5} & = w \\[-1.5ex] c_{a_5} & = 1 \end{aligned}$};
\path[->] (U) edge[simpleedge,out=40,in=-180] (x);
\path[->] (U) edge[simpleedge,out=40,in=120] (y);
\oldhyperedge{y}{y,U}{U}{0.48}{0.5}{t};

\node[anchor=west] at (7, -1.5) {$\begin{aligned} 
S & = [U] \\[-0.8ex]
n &= 3 \\[-0.8ex]
F & = [a_4]\\[-0.8ex]
U & = \Find{u} = \Find{v} \\[-0.8ex]
& = \Find{w} 
\end{aligned}$ };
\end{tikzpicture}
\end{small}
\end{center}
Then $a_4$ is popped from $F$, and the loop from~\lineref{scc:begin_edge_loop2} to~\lineref{scc:end_edge_loop2} iterates over $H(a_4) = \{x,y\}$. Suppose that $x$ is treated first. Then $\Call{Visit}{x}$ is called. During its execution, at Line~\lineref{scc:end_node_loop}, the state is:
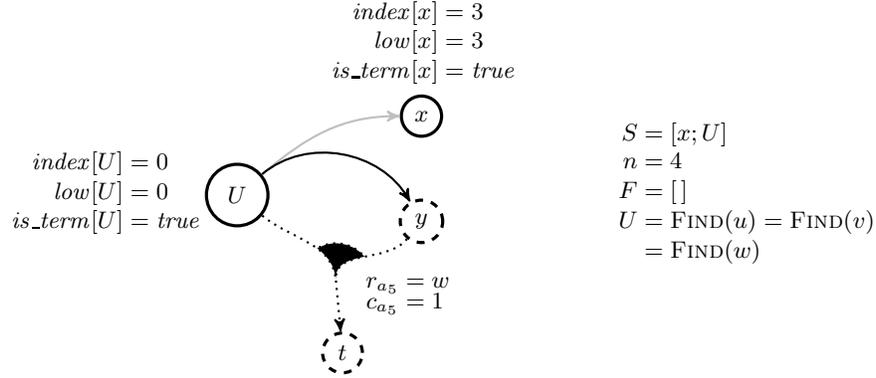
\begin{center}
\begin{small}
\begin{tikzpicture}[>=stealth',scale=0.7, vertex/.style={circle,draw=black,very thick,minimum size=2ex}, hyperedge/.style={draw=black,thick,dotted}, simpleedge/.style={draw=black,thick}]
\node [vertex,very thick,text width=0.48cm,text centered] (U) at (0,-1.5) {$U$} node[node distance=13ex,left of=U] {
$\begin{aligned} 
\tindex[U] &= 0  \\[-0.8ex]
\troot[U] &= 0 \\[-0.8ex]
\ismax[U] &= \True 
\end{aligned}$
};
\node [vertex] (x) at (3.5,0) {$x$} node[node distance=7ex,above of=x] {
$\begin{aligned} 
\tindex[x] &= 3  \\[-0.8ex]
\troot[x] &= 3 \\[-0.8ex]
\ismax[x] &= \True 
\end{aligned}$
};
\node [vertex,dashed] (y) at (3.5,-2) {$y$};
\node [vertex,dashed] (t) at (2,-4.5) {$t$};

\node at (3.2,-3.4) {$\begin{aligned} r_{a_5} & = w \\[-1.5ex] c_{a_5} & = 1 \end{aligned}$};
\path[->] (U) edge[simpleedge,out=40,in=-180,lightgray] (x);
\path[->] (U) edge[simpleedge,out=40,in=120] (y);
\oldhyperedge{y}{y,U}{U}{0.48}{0.5}{t};

\node[anchor=west] at (7, -1.5) {$\begin{aligned} 
S & = [x;U] \\[-0.8ex]
n & = 4 \\[-0.8ex]
F & = \emptystack\\[-0.8ex]
U & = \Find{u} = \Find{v} \\[-0.8ex]
& = \Find{w} 
\end{aligned}$ };
\end{tikzpicture}
\end{small}
\end{center}
Since $F$ is empty, the loop from Lines~\lineref{scc:begin_edge_loop} to~\lineref{scc:end_edge_loop} is not executed. At Line~\lineref{scc:begin2}, $\troot[x] = \tindex[x]$ and $\ismax[x] = \True$, so that a trivial vertex merging step is performed, only on $x$, since it is the top element of $S$. After Line~\lineref{scc:index_redef}, it can be verified that $S = [x; U]$, $\tindex[x] = 3$ and $F = \emptystack$. Therefore, the goto statement at Line~\lineref{scc:goto} is not executed. It follows that the loop from Lines~\lineref{scc:begin_non_max_scc_loop} to~\lineref{scc:end_non_max_scc_loop} is executed, and after that, the state is:

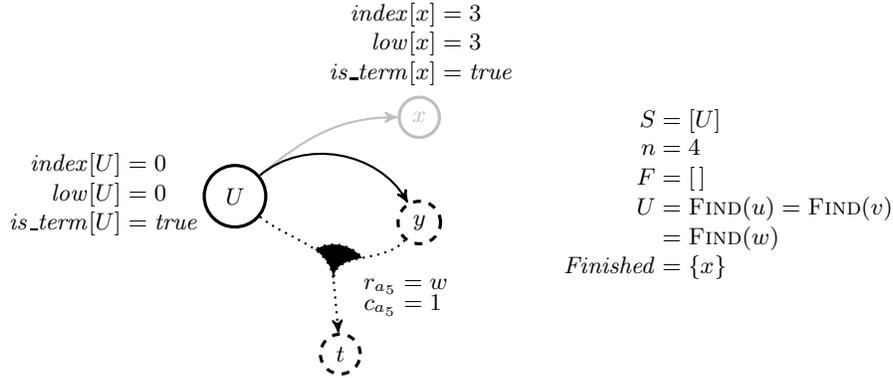
\begin{center}
\begin{small}
\begin{tikzpicture}[>=stealth',scale=0.7, vertex/.style={circle,draw=black,very thick,minimum size=2ex}, hyperedge/.style={draw=black,thick,dotted}, simpleedge/.style={draw=black,thick}]
\node [vertex,very thick,text width=0.48cm,text centered] (U) at (0,-1.5) {$U$} node[node distance=13ex,left of=U] {
$\begin{aligned} 
\tindex[U] &= 0  \\[-0.8ex]
\troot[U] &= 0 \\[-0.8ex]
\ismax[U] &= \True 
\end{aligned}$
};
\node [vertex,lightgray] (x) at (3.5,0) {$x$} node[node distance=7ex,above of=x] {
$\begin{aligned} 
\tindex[x] &= 3  \\[-0.8ex]
\troot[x] &= 3 \\[-0.8ex]
\ismax[x] &= \True 
\end{aligned}$
};
\node [vertex,dashed] (y) at (3.5,-2) {$y$};
\node [vertex,dashed] (t) at (2,-4.5) {$t$};

\node at (3.2,-3.4) {$\begin{aligned} r_{a_5} & = w \\[-1.5ex] c_{a_5} & = 1 \end{aligned}$};
\path[->] (U) edge[simpleedge,out=40,in=-180,lightgray] (x);
\path[->] (U) edge[simpleedge,out=40,in=120] (y);
\oldhyperedge{y}{y,U}{U}{0.48}{0.5}{t};

\node[anchor=west] at (6, -1.5) {$\begin{aligned} 
S & = [U] \\[-0.8ex]
n & = 4 \\[-0.8ex]
F & = \emptystack\\[-0.8ex]
U & = \Find{u} = \Find{v} \\[-0.8ex] 
& = \Find{w} \\[-0.8ex]
\incomp & = \{ x \}
\end{aligned}$ };
\end{tikzpicture}
\end{small}
\end{center}
After the termination of $\Call{Visit}{x}$, since $x \in \incomp$, $\ismax[U]$ is set to $\False$. After that, $\Call{Visit}{y}$ is called, and at Line~\lineref{scc:end_node_loop}, it can be checked that $c_{a_5}$ has been incremented to $2 = \card{T(a_5)}$ because $R_{a_5} = \Find{r_{a_5}} = \Find{w} = U$ and $U \in S$. Therefore, $a_5$ is pushed to $F_U$, and the state is:
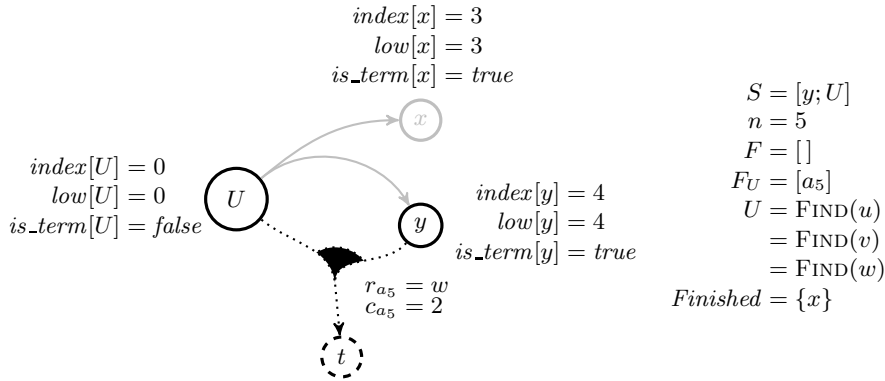
\begin{center}
\begin{small}
\begin{tikzpicture}[>=stealth',scale=0.7, vertex/.style={circle,draw=black,very thick,minimum size=2ex}, hyperedge/.style={draw=black,thick,dotted}, simpleedge/.style={draw=black,thick}]
\node [vertex,very thick,text width=0.48cm,text centered] (U) at (0,-1.5) {$U$} node[node distance=13ex,left of=U] {
$\begin{aligned} 
\tindex[U] &= 0  \\[-0.8ex]
\troot[U] &= 0 \\[-0.8ex]
\ismax[U] &= \False
\end{aligned}$
};
\node [vertex,lightgray] (x) at (3.5,0) {$x$} node[node distance=7ex,above of=x] {
$\begin{aligned} 
\tindex[x] &= 3  \\[-0.8ex]
\troot[x] &= 3 \\[-0.8ex]
\ismax[x] &= \True 
\end{aligned}$
};
\node [vertex] (y) at (3.5,-2) {$y$} node[node distance=12ex,right of=y] {
$\begin{aligned} 
\tindex[y] &= 4  \\[-0.8ex]
\troot[y] &= 4 \\[-0.8ex]
\ismax[y] &= \True 
\end{aligned}$
};
\node [vertex,dashed] (t) at (2,-4.5) {$t$};

\path[->] (U) edge[simpleedge,out=40,in=-180,lightgray] (x);
\path[->] (U) edge[simpleedge,out=40,in=120,lightgray] (y);
\oldhyperedge{y}{y,U}{U}{0.48}{0.5}{t};

\node at (3.2,-3.4) {$\begin{aligned} r_{a_5} & = w \\[-1.5ex] c_{a_5} & = 2 \end{aligned}$};

\node[anchor=west] at (8, -1.5) {$\begin{aligned} 
S & = [y;U] \\[-0.8ex]
n & = 5 \\[-0.8ex]
F & = \emptystack\\[-0.8ex]
F_U & = [a_5]\\[-0.8ex]
U & = \Find{u} \\[-0.8ex]
& = \Find{v} \\[-0.8ex] 
& = \Find{w} \\[-0.8ex]
\incomp & = \{ x \}
\end{aligned}$ };
\end{tikzpicture}
\end{small}
\end{center}
As for the vertex $x$, $\Call{Visit}{y}$ terminates by 
popping $y$ from $S$ and adding it to $\incomp$. Back to the execution of $\Call{Visit}{U}$, at Line~\lineref{scc:begin2}, the state is:
\begin{center}
\begin{small}
\begin{tikzpicture}[>=stealth',scale=0.7, vertex/.style={circle,draw=black,very thick,minimum size=2ex}, hyperedge/.style={draw=black,thick,dotted}, simpleedge/.style={draw=black,thick}]
\node [vertex,very thick,text width=0.48cm,text centered] (U) at (0,-1.5) {$U$} node[node distance=13ex,left of=U] {
$\begin{aligned} 
\tindex[U] &= 0  \\[-0.8ex]
\troot[U] &= 0 \\[-0.8ex]
\ismax[U] &= \False
\end{aligned}$
};
\node [vertex,lightgray] (x) at (3.5,0) {$x$} node[node distance=7ex,above of=x] {
$\begin{aligned} 
\tindex[x] &= 3  \\[-0.8ex]
\troot[x] &= 3 \\[-0.8ex]
\ismax[x] &= \True 
\end{aligned}$
};
\node [vertex,lightgray] (y) at (3.5,-2) {$y$} node[node distance=12ex,right of=y] {
$\begin{aligned} 
\tindex[y] &= 4  \\[-0.8ex]
\troot[y] &= 4 \\[-0.8ex]
\ismax[y] &= \True 
\end{aligned}$
};
\node [vertex,dashed] (t) at (2,-4.5) {$t$};

\path[->] (U) edge[simpleedge,out=40,in=-180,lightgray] (x);
\path[->] (U) edge[simpleedge,out=40,in=120,lightgray] (y);
\oldhyperedge{y}{y,U}{U}{0.48}{0.5}{t};

\node at (3.2,-3.4) {$\begin{aligned} r_{a_5} & = w \\[-1.5ex] c_{a_5} & = 2 \end{aligned}$};

\node[anchor=west] at (8, -1.5) {$\begin{aligned} 
S & = [U] \\[-0.8ex]
n & = 5 \\[-0.8ex]
F & = \emptystack\\[-0.8ex]
F_U & = [a_5]\\[-0.8ex]
U & = \Find{u}\\[-0.8ex]
& = \Find{v} \\[-0.8ex] 
& = \Find{w} \\[-0.8ex]
\incomp & = \{ y, x \}
\end{aligned}$ };
\end{tikzpicture}
\end{small}
\end{center}
While $\troot[U] = \tindex[U]$, $\ismax[U]$ is equal to $\False$, so that no vertex merging loop is performed on $U$. Therefore, $a_5$ is not popped from $F_U$. 
Nevertheless, the loop from Lines~\lineref{scc:begin_non_max_scc_loop} to~\lineref{scc:end_non_max_scc_loop} is executed, and after that, $\Call{Visit}{u}$ is terminated in the state:
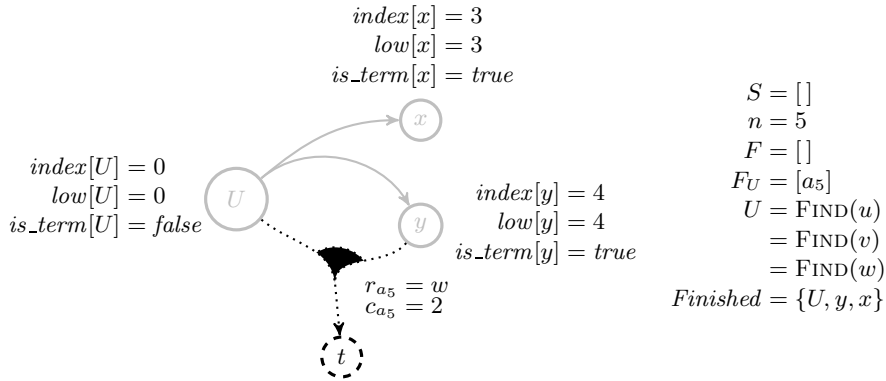
\begin{center}
\begin{small}
\begin{tikzpicture}[>=stealth',scale=0.7, vertex/.style={circle,draw=black,very thick,minimum size=2ex}, hyperedge/.style={draw=black,thick,dotted}, simpleedge/.style={draw=black,thick}]
\node [vertex,very thick,text width=0.48cm,text centered, lightgray] (U) at (0,-1.5) {$U$} node[node distance=13ex,left of=U] {
$\begin{aligned} 
\tindex[U] &= 0  \\[-0.8ex]
\troot[U] &= 0 \\[-0.8ex]
\ismax[U] &= \False
\end{aligned}$
};
\node [vertex,lightgray] (x) at (3.5,0) {$x$} node[node distance=7ex,above of=x] {
$\begin{aligned} 
\tindex[x] &= 3  \\[-0.8ex]
\troot[x] &= 3 \\[-0.8ex]
\ismax[x] &= \True 
\end{aligned}$
};
\node [vertex,lightgray] (y) at (3.5,-2) {$y$} node[node distance=12ex,right of=y] {
$\begin{aligned} 
\tindex[y] &= 4  \\[-0.8ex]
\troot[y] &= 4 \\[-0.8ex]
\ismax[y] &= \True 
\end{aligned}$
};
\node [vertex,dashed] (t) at (2,-4.5) {$t$};

\path[->] (U) edge[simpleedge,out=40,in=-180,lightgray] (x);
\path[->] (U) edge[simpleedge,out=40,in=120,lightgray] (y);
\oldhyperedge{y}{y,U}{U}{0.48}{0.5}{t};

\node at (3.2,-3.4) {$\begin{aligned} r_{a_5} & = w \\[-1.5ex] c_{a_5} & = 2 \end{aligned}$};

\node[anchor=west] at (8, -1.5) {$\begin{aligned} 
S & = \emptystack \\[-0.8ex]
n & = 5 \\[-0.8ex]
F & = \emptystack\\[-0.8ex]
F_U & = [a_5]\\[-0.8ex]
U & = \Find{u}\\[-0.8ex]
& = \Find{v} \\[-0.8ex] 
& = \Find{w} \\[-0.8ex]
\incomp & = \{ U, y, x \}
\end{aligned}$ };
\end{tikzpicture}
\end{small}
\end{center}
Finally, $\Call{Visit}{t}$ is called from \Call{TerminalScc}{} at Line~\lineref{scc:visit_call}. It can be verified that a trivial vertex merging loop is performed on $t$ only. 
After that, $t$ is placed into $\incomp$. Therefore, the final state of \Call{TerminalScc}{} is:
\begin{center}
\begin{small}
\begin{tikzpicture}[>=stealth',scale=0.7, vertex/.style={circle,draw=black,very thick,minimum size=2ex}, hyperedge/.style={draw=black,thick,dotted}, simpleedge/.style={draw=black,thick}]
\node [vertex,very thick,text width=0.48cm,text centered, lightgray] (U) at (0,-1.5) {$U$} node[node distance=13ex,left of=U] {
$\begin{aligned} 
\tindex[U] &= 0  \\[-0.8ex]
\troot[U] &= 0 \\[-0.8ex]
\ismax[U] &= \False
\end{aligned}$
};
\node [vertex,lightgray] (x) at (3.5,0) {$x$} node[node distance=7ex,above of=x] {
$\begin{aligned} 
\tindex[x] &= 3  \\[-0.8ex]
\troot[x] &= 3 \\[-0.8ex]
\ismax[x] &= \True 
\end{aligned}$
};
\node [vertex,lightgray] (y) at (3.5,-2) {$y$} node[node distance=12ex,right of=y] {
$\begin{aligned} 
\tindex[y] &= 4  \\[-0.8ex]
\troot[y] &= 4 \\[-0.8ex]
\ismax[y] &= \True 
\end{aligned}$
};
\node [vertex,lightgray] (t) at (2,-4.5) {$t$} node[node distance=7ex,below of=t] {
$\begin{aligned} 
\tindex[t] &= 5 \\[-0.8ex]
\troot[t] &= 5 \\[-0.8ex]
\ismax[t] &= \True 
\end{aligned}$
};;

\path[->] (U) edge[simpleedge,out=40,in=-180,lightgray] (x);
\path[->] (U) edge[simpleedge,out=40,in=120,lightgray] (y);
\oldhyperedge{y}{y,U}{U}{0.48}{0.5}{t};

\node at (3.2,-3.4) {$\begin{aligned} r_{a_5} & = w \\[-1.5ex] c_{a_5} & = 2 \end{aligned}$};

\node[anchor=west] at (7.8, -1.5) {$\begin{aligned} 
S & = \emptystack \\[-0.8ex]
n & = 6 \\[-0.8ex]
F_U & = [a_5]\\[-0.8ex]
U & = \Find{u}\\[-0.8ex]
& = \Find{v} \\[-0.8ex]
& = \Find{w} \\[-0.8ex]
\incomp & = \{ t, U, y, x \}
\end{aligned}$ };
\end{tikzpicture}
\end{small}
\end{center}
As $\ismax[x] = \ismax[y] = \ismax[t] = \True$ and $\ismax[\Find{z}] = \False$ for $z = u, v, w$, there are three terminal \scc{}s, given by the sets:
\begin{align*}
\{ z \mid \Find{z} = x \} & = \{x\}, \\
\{ z \mid \Find{z} = y \} & = \{y\}, \\
\{ z \mid \Find{z} = t \} & = \{t\}.
\end{align*}

\section{Proof of Theorem~\ref{th:correctness}}\label{sec:correctness_proof}

The correctness proof of the algorithm \Call{TerminalScc}{} turns out to be harder than for algorithms on directed graphs such as Tarjan's one~\cite{Tarjan72}, due to the complexity of the invariants which arise in the former algorithm. That is why we propose to show the correctness of two intermediary algorithms, named \Call{TerminalScc2}{} (Figure~\ref{fig:maxscc2}) and \Call{TerminalScc3}{} (Figure~\ref{fig:maxscc3}), and then to prove that they are equivalent to \Call{TerminalScc}{}.

\begin{figure}[t]
\begin{scriptsize}
\begin{minipage}[t]{0.41\textwidth}
\vspace{0pt}
\begin{algorithmic}[1]
\Function {TerminalScc2}{$\VV,A$}
  \State $n \gets 0$, $S \gets \emptystack$, $\incomp \gets \emptyset$
  \ForAll{$a \in A$} 
    \State $\collected_a \gets \False$
  \EndFor
  \ForAll{$u \in A$} 
    \State $\tindex[u] \gets \Nil$
    \State $\troot[u] \gets \Nil$
    \State \Call{Makeset}{$u$}
  \EndFor
  \ForAll{$u \in \VV$}
    \If{$\tindex[u] = \Nil$} 
      \State \Call{Visit2}{$u$}
	\EndIf
  \EndFor
\EndFunction
\Statex
\Statex
\Statex
\Function {Visit2}{$u$} 
  \State local $U \gets \Call{Find}{u}$\label{scc2:find1}, local $F \gets \emptyset$\label{scc2:begin_atom1}
  \State $\tindex[U] \gets n$, $\troot[U] \gets n$
  \State $n \gets n+1$
  \State $\ismax[U] \gets \True$
  \State push $U$ on the stack $S$\label{scc2:push1}
  \State local $\single \gets \True$
  \State $F \gets \{ a \in A \mid T(a) = \{ u  \}  \}$\label{scc2:f_assign}
  \ForAll{$a \in F$} 
    \State $\collected_a \gets \True$
  \EndFor\label{scc2:end_atom1}
  \algstore{scc_break}
  \end{algorithmic}
  \end{minipage}
  \hfill
  \begin{minipage}[t]{0.58\textwidth}
  \vspace{0pt}
  \begin{algorithmic}[1]
  \algrestore{scc_break}
  \While{$F$ is not empty} \label{scc2:begin_edge_loop}
    \State pop $a$ from $F$
    \ForAll{$w \in H(a)$}
      \State local $W \gets \Call{Find}{w}$\label{scc2:find3}
      \sIf{$\tindex[W] = \Nil$} \Call{Visit2}{$w$}\label{scc2:rec_call}
      \If{$W \in \incomp$}
	\State $\ismax[U] \gets \False$
      \Else
	\State $\troot[U] \gets \min(\troot[U],\troot[W])$
	\State $\ismax[U] \gets \ismax[U] \And \ismax[W]$
      \EndIf
    \EndFor 
  \EndWhile\label{scc2:end_edge_loop}
  \If{$\troot[U] = \tindex[U]$} \label{scc2:is_root}
    \If{$\ismax[U] = \True$}
      \State local $i \gets \tindex[U]$\label{scc2:begin_node_merging}\label{scc2:begin_atom2}
      \State pop $V$ from $S$\label{scc2:pop1}
      \While{$\tindex[V] > i$} \label{scc2:begin_node_merging_loop}
	\State $\single \gets \False$
	\State $U \gets \Call{Merge}{U, V}$\label{scc2:merge}
	\State pop $V$ from $S$\label{scc2:pop2}
      \EndWhile \label{scc2:end_node_merging_loop}
      \State push $U$ on $S$\label{scc2:push2}
      \State $F \gets 
		\biggl\{ a \in A \Bigm| \begin{aligned}
		 	      & \collected_a = \False, \\[-1ex]
		 	      & \forall x \in T(a),\Call{Find}{x} = U 
		 	    \end{aligned} \biggr\}$ \label{scc2:f_assign2}
      \sForAll{$a \in F$} $\collected_a \gets \True$\label{scc2:end_atom2}\label{scc2:end_node_merging}
      
      \If{$\single = \False$} \label{scc2:begin_not_single}
	\State $n \gets i$, $\tindex[U] \gets n$, $n \gets n+1$ \label{scc2:nredefined}
	\State $\single \gets \True$, go to Line~\lineref{scc2:begin_edge_loop}\label{scc2:goto}
       \label{scc2:end_not_single}
      \EndIf
    \EndIf
    \Repeat \label{scc2:begin_non_max_scc_loop}\label{scc2:begin_atom3}
      \State pop $V$ from $S$, add $V$ to $\incomp$\label{scc2:pop3}
    \Until{$\tindex[V] = \tindex[U]$} \label{scc2:end_non_max_scc_loop}\label{scc2:end_atom3}
  \EndIf
\EndFunction
\end{algorithmic}
\end{minipage}
\end{scriptsize}
\caption{First intermediary form of our algorithm computing the terminal \scc{}s}\label{fig:maxscc2}
\end{figure}

The main difference between the first intermediary form and \Call{TerminalScc}{} is that it does not use auxiliary data associated to the hyperarcs to determine which ones are added to the digraph $\graph(\HH_\cur)$ after a vertex merging step. Instead, the stack $F$ is directly filled with the right hyperarcs (Lines~\lineref{scc2:f_assign} and~\lineref{scc2:f_assign2}). Besides, a boolean $\single$ is used to determine whether a vertex merging step has been executed. The notion of \emph{vertex merging step} is refined: it now refers to the execution of the instructions between Lines~\lineref{scc2:begin_node_merging} and~\lineref{scc2:end_node_merging} in which the boolean $\single$ is set to $\False$. 

For the sake of simplicity, we will suppose that sequences of assignment or stack manipulations are executed atomically. For instance, the sequences of instructions located in the blocks from Lines~\lineref{scc2:begin_atom1} and~\lineref{scc2:end_atom1}, or from Lines~\lineref{scc2:begin_atom2} and~\lineref{scc2:end_atom2}, and at from Lines~\lineref{scc2:begin_atom3} to~\lineref{scc2:end_atom3}, are considered as elementary instructions. Under this assumption, intermediate complex invariants do not have to be considered.

We first begin with very simple invariants:
\begin{invariant}\label{inv:tindex}
Let $U$ be a vertex of the current hypergraph $\HH_\cur$. Then $\tindex[U]$ is defined if, and only if, $\tindex[u]$ is defined for all $u \in \VV$ such that $\Find{u} = U$.
\end{invariant}

\begin{proof}
It can be shown by induction on the number of vertex merging steps which has been performed on $U$. 

In the basis case, there is a unique element $u \in \VV$ such that $\Find{u} = U$. Besides, $U = u$, so that the statement is trivial. 

After a merging step yielding the vertex $U$, we necessarily have $\tindex[U] \neq \Nil$. Moreover, all the vertices $V$ which has been merged into $U$ satisfied $\tindex[V] \neq \Nil$ because they were stored in the stack $S$. Applying the induction hypothesis terminates the proof.
\end{proof}

\begin{invariant}\label{inv:comporstack}
Let $u \in \VV$. When $\tindex[u]$ is defined, then $\Call{Find}{u}$ belongs either to the stack $S$, or to the set $\incomp$ (both cases cannot happen simultaneously).
\end{invariant}

\begin{proof}
Initially, $\Find{u} = u$, and once $\tindex[u]$ is defined, $\Find{u}$ is pushed on $S$ (Line~\lineref{scc2:push1}). Naturally, $u \not \in \incomp$, because otherwise, $\tindex[u]$ would have been defined before (see the condition Line~\lineref{scc2:end_non_max_scc_loop}). After that, $U = \Find{u}$ can be popped from $S$ at three possible locations:
\begin{itemize}[\textbullet]
\item at Lines~\lineref{scc2:pop1} or~\lineref{scc2:pop2}, in which case $U$ is transformed into a vertex $U'$ which is immediately pushed on the stack $S$ at Line~\lineref{scc2:push2}. Since after that, $\Find{u} = U'$, the property $\Find{u} \in S$ still holds.
\item at Line~\lineref{scc2:pop3}, in which case it is directly appended to the set $\incomp$.\qedhere
\end{itemize}
\end{proof}

\begin{invariant}\label{inv:incomp}
The set $\incomp$ is always growing. 
\end{invariant}

\begin{proof}
Once an element is added to $\incomp$, it is never removed from it nor merged into another vertex (the function \Call{Merge}{} is always called on elements immediately popped from the stack $S$).
\end{proof}

\begin{proposition}\label{prop:maxscc2}
After the algorithm $\Call{TerminalScc2}{\HH}$ terminates, the sets $\{ v \in \VV \mid \Call{Find}{v} = U \text{ and } \ismax[U] = \True \}$ are precisely the terminal \scc{}s of $\HH$.
\end{proposition}

\begin{proof}
We prove the whole statement by induction on the number of vertex merging steps.

\textit{Basis Case.} First, suppose that the hypergraph $\HH$ is such that no vertices are merged during the execution of \Call{TerminalScc2}{$\HH$}, \ie\ the vertex merging loop (from Lines~\lineref{scc2:begin_node_merging_loop} to~\lineref{scc2:end_node_merging_loop}) is never executed. Then the boolean $\single$ is always set to $\True$, so that $n$ is never redefined to $i+1$ (Line~\lineref{scc2:nredefined}), and there is no back edge to Line~\lineref{scc2:begin_edge_loop} in the control-flow graph. It follows that removing all the lines between Lines~\lineref{scc2:begin_node_merging} to~\lineref{scc2:end_not_single} does not change the behavior of the algorithm. Besides, since the function \Call{Merge}{} is never called, $\Call{Find}{u}$ always coincides with $u$. Finally, at Line~\lineref{scc2:f_assign}, $F$ is precisely assigned to the set of simple hyperarcs leaving $u$ in $\HH$, so that the loop from Lines~\lineref{scc2:begin_edge_loop} to~\lineref{scc2:end_edge_loop} iterates on the successors of $u$ in $\graph(\HH)$. As a consequence, the algorithm \Call{TerminalScc2}{$\HH$} behaves exactly like \Call{TerminalScc2}{$\graph(\HH)$}. Moreover, under our assumption, the terminal \scc{}s of $\graph(\HH)$ are all reduced to singletons (otherwise, the loop from Lines~\lineref{scc2:begin_node_merging_loop} to~\lineref{scc2:end_node_merging_loop} would be executed, and some vertices would be merged). Therefore, by Proposition~\ref{prop:terminal_scc}, the statement in Proposition~\ref{prop:maxscc2} holds.

\textit{Inductive Case.} Suppose that the vertex merging loop is executed at least once, and that its first execution happens during the execution of, say, \Call{Visit2}{$x$}. Consider the state of the algorithm at Line~\lineref{scc2:begin_node_merging} just before the execution of the first occurrence of the vertex merging step. Until that point, $\Find{v}$ is still equal to $v$ for all vertices $v \in \VV$, so that the execution of \Call{TerminalScc2}{$\HH$} coincides with the execution of \Call{TerminalScc2}{$\graph(\HH)$}. Consequently, if $C$ is the set formed by the vertices $y$ located above $x$ in the stack $S$ (including $x$), $C$ forms a terminal \scc{} of $\graph(\HH)$. In particular, the elements of $C$ are located in a same \scc{} of the hypergraph $\HH$.

Consider the hypergraph $\HH'$ obtained by merging the elements of $C$ in the hypergraph $(\VV,A \setminus \{ a \mid \exists y \in C \text{ s.t. } T(a) = \{ y \} \})$, and let $X$ be the resulting vertex. For now, we may add a hypergraph as last argument of the functions \Call{Visit2}{}, \Call{Find}{}, \etc{}, to distinguish their execution in the context of the call to \Call{TerminalScc2}{$\HH$} or \Call{TerminalScc2}{$\HH'$}. We make the following observations:
\begin{itemize}[\textbullet]
\item the vertex $x$ is the first element of the component $C$ to be visited during the execution of \Call{TerminalScc2}{$\HH$}. It follows that the execution of \Call{TerminalScc2}{$\HH$} until the call to \Call{Visit2}{$x,\HH$} coincides with the execution of \Call{TerminalScc2}{$\HH'$} until the call to \Call{Visit2}{$X,\HH'$}.

\item besides, during the execution of \Call{Visit2}{$x,\HH$}, the execution of the loop from Lines~\lineref{scc2:begin_edge_loop} to~\lineref{scc2:end_edge_loop} only has a local impact, \ie\ on the $\ismax[y]$, $\tindex[y]$, or $\troot[y]$ for $y \in C$, and not on any information relative to other vertices. Indeed, we claim that the set of the vertices $y$ on which \Call{Visit2}{} is called during the execution of the loop is exactly $C \setminus \{ x  \}$. First, for all $y \in C \setminus \{ x \}$, \Call{Visit2}{$y$} has necessarily been executed \emph{after} Line~\lineref{scc2:begin_edge_loop} (otherwise, by Invariant~\ref{inv:comporstack}, $y$ would be either below $x$ in the stack $S$, or in $\incomp$). Conversely, suppose that after Line~\lineref{scc2:begin_edge_loop}, there is a call to \Call{Visit2}{$t$} with $t \not \in C$. By Invariant~\ref{inv:comporstack}, $t$ belongs to $\incomp$, so that for one of the vertices $w$ examined in the loop, either $w \in \incomp$ or $\ismax[w] = \False$ after the call to \Call{Visit2}{$w$}. Hence $\ismax[x]$ should be $\False$, which contradicts our assumptions.

\item finally, from the execution of Line~\lineref{scc2:goto} during the call to \Call{Visit2}{$x,\HH$}, our algorithm behaves exactly as \Call{TerminalScc2}{$\HH'$} from the execution of Line~\lineref{scc2:begin_edge_loop} in \Call{Visit2}{$X,\HH'$}. Indeed, $\tindex[X]$ is equal to $i$, and the latter is equal to $n-1$. Similarly, for all $y \in C$, $\troot[y] = i$ and $\ismax[y] = \True$. The vertex $X$ being equal to one of the $y \in C$, we also have $\troot[X] = i$ and $\ismax[X] = \True$. Moreover, $X$ is the top element of $S$. 

Furthermore, it can be verified that at Line~\lineref{scc2:f_assign2}, the set $F$ contains exactly all the hyperarcs of $A$ which generate the simple hyperarcs leaving $X$ in $\HH'$: they are exactly characterized by
\begin{align*}
& \Find{z,\HH} = X \text{ for all } z \in T(a), \text{ and }T(a) \neq \{y\} \text{ for all } y \in C  \\
\Longleftrightarrow{} & \Find{z,\HH} = X \text{ for all } z \in T(a), \text{ and } \collected_a = \False 
\end{align*}
since at that Line~\lineref{scc2:f_assign2}, a hyperarc $a$ satisfies $\collected_a = \True$ if, and only if, $T(a)$ is reduced to a singleton $\{t\}$ such that $\tindex[t]$ is defined. 

Finally, for all vertices $y \in C$, $\Call{Find}{y,\HH}$ can be equivalently replaced by $\Call{Find}{X,\HH'}$.
\end{itemize}
As a consequence, \Call{TerminalScc2}{$\HH$} and \Call{TerminalScc2}{$\HH'$} return the same result. Both functions perform the same union-find operations, except the first the vertex merging step executed by \Call{TerminalScc2}{$\HH$} on $C$.

Let $f$ be the function which maps all vertices $y \in C$ to $X$, and any other vertex to itself. We claim that $\HH'$ and $f(\HH)$ have the same reachability graph, \ie\ $\reach_{\HH'}$ and $\reach_{f(\HH)}$ are identical relations. Indeed, the two hypergraphs only differ on the images of the hyperarcs $a \in A$ such that $T(a) = \{y\}$ for some $y \in C$. For such hyperarcs, we have $H(a) \subseteq C$, because otherwise, $\ismax[x]$ would have been set to $\False$ (\ie\ the component $C$ would not be terminal). It follows that their are mapped to the cycle $(\{X\},\{X\})$ by $f$, so that $\HH'$ and $f(\HH)$ clearly have the same reachability graph. In particular, they have the same terminal \scc{}s.

Finally, since the elements of $C$ are in a same \scc{} of $\HH$, Proposition~\ref{prop:collapse} shows that the function $f$ induces a one-to-one correspondence between the \scc{}s of $\HH$ and the \scc{}s of $f(\HH)$:
\begin{align*}
D & \longmapsto f(D) \\
(D' \setminus \{ X \}) \cup C & \longmapsfrom D' && \text{if }X \in D' \\
D' & \longmapsfrom D' && \text{otherwise}.
\end{align*}
The action of the function $f$ exactly corresponds to the vertex merging step performed on $C$. Since by induction hypothesis, \Call{TerminalScc2}{$\HH'$} determines the terminal \scc{}s in $f(\HH)$, it follows that Proposition~\ref{prop:maxscc2} holds.
\end{proof}

\begin{figure}[t]
\begin{scriptsize}
\begin{minipage}[t]{0.44\textwidth}
\vspace{0pt}
\begin{algorithmic}[1]
\Function {TerminalScc3}{$\VV,A$}
  \State $n \gets 0$, $S \gets \emptystack$, $\incomp \gets \emptyset$
  \ForAll{$a \in A$} \tikz[remember picture, baseline]{\coordinate (t4);}
    \State \tikz[remember picture, baseline]{\coordinate (l4);}$r_a \gets \Nil$, $c_a \gets 0$\tikz[remember picture, baseline]{\coordinate (r4);}\tikz[remember picture, baseline]{\coordinate (b4);}
    \State $\collected_a \gets \False$\label{scc3:collected_init}
  \EndFor
  \ForAll{$u \in \VV$} 
    \State $\tindex[u] \gets \Nil$
    \State $\troot[u] \gets \Nil$
    \State \Call{Makeset}{$u$}, $F_u \gets \emptystack$
  \EndFor
  \ForAll{$u \in \VV$}
    \If{$\tindex[u] = \Nil$} 
      \State \Call{Visit3}{$u$}\label{scc3:init_call}
    \EndIf
  \EndFor
\EndFunction
\Statex
\Function {Visit3}{$u$} 
  \State local $U \gets \Call{Find}{u}$\label{scc3:find1}, local $F \gets \emptystack$\label{scc3:begin_atom1}
  \State $\tindex[U] \gets n$, $\troot[U] \gets n$
  \State $n \gets n+1$\label{scc3:troot_def}
  \State $\ismax[U] \gets \True$
  \State push $U$ on the stack $S$
  \ForAll{$a \in A_u$}\label{scc3:begin_node_loop}
    \sIf{$\card{T(a)} = 1$} push $a$ on $F$\label{scc3:f_push}
    \fElse
      \sIf{$r_a = \Nil$} $r_a \gets u$\tikz[remember picture, baseline]{\coordinate (t);}
      \State \tikz[remember picture, baseline]{\coordinate (l);}local $R_a \gets \Call{Find}{r_a}$\label{scc3:find2}
      \If{$R_a$ appears in $S$}\label{scc3:root_reach}
	\State $c_a \gets c_a + 1$\label{scc3:counter_increment}
	\If{$c_a = \card{T(a)}$}\label{scc3:counter_reach}
	  \State push $a$ on the stack $F_{R_a}$\label{scc3:stack_edge}\tikz[remember picture, baseline]{\coordinate (r);}\label{scc3:root_def}
	\EndIf
      \EndIf\tikz[remember picture, baseline]{\coordinate (b);}
    \EndIf
  \EndFor
  \ForAll{$a \in F$} \label{scc3:collected1_begin}
    \State $\collected_a \gets \True$
  \EndFor\label{scc3:collected1_end}\label{scc3:end_atom1}
\algstore{scc_break}
\end{algorithmic}
\end{minipage}
\hfill
\begin{minipage}[t]{0.55\textwidth}
\vspace{0pt}
\begin{algorithmic}[1]
\algrestore{scc_break}
  \While{$F$ is not empty} \label{scc3:begin_edge_loop}
    \State pop $a$ from $F$
    \ForAll{$w \in H(a)$}
      \State local $W \gets \Call{Find}{w}$\label{scc3:find3}
      \sIf{$\troot[W] = \Nil$} \Call{Visit3}{$w$}\label{scc3:rec_call}
      \If{$W \in \incomp$}
	\State $\ismax[U] \gets \False$
      \Else
	\State $\troot[U] \gets \min(\troot[U],\troot[W])$
	\State $\ismax[U] \!\! \gets \!\! \ismax[U] \! \And \! \ismax[W]$
      \EndIf
    \EndFor
  \EndWhile \label{scc3:end_edge_loop}
  \If{$\troot[U] = \tindex[U]$} 
    \If{$\ismax[U] = \True$}\tikz[remember picture, baseline]{\coordinate (t2);}
      \State local $i \gets \tindex[U]$\label{scc3:begin_node_merging}
      \State \tikz[remember picture, baseline]{\coordinate (l2);}pop each $a \in F_U$ and push it on $F$\tikz[remember picture, baseline]{\coordinate (r2);}\tikz[remember picture, baseline]{\coordinate (b2);}\label{scc3:push_on_fprime1}\label{scc3:begin_atom2}
      \State pop $V$ from $S$
      \While{$\tindex[V] > i$} \tikz[remember picture, baseline]{\coordinate (t3);}\label{scc3:begin_node_merging_loop}
	\State \tikz[remember picture, baseline]{\coordinate (l3);}pop each $a \in F_V$ and push it on $F$\tikz[remember picture, baseline]{\coordinate (r3);}\tikz[remember picture, baseline]{\coordinate (b3);}\label{scc3:push_on_fprime2}
	\State $U \gets \Call{Merge}{U, V}$\label{scc3:merge}
	\State pop $V$ from $S$
      \EndWhile \label{scc3:end_node_merging_loop}\label{scc3:end_node_merging}
      \State $\tindex[U] \gets i$, push $U$ on $S$
      \State $F \gets \biggl\{\! a \in A \!\Bigm|\! \begin{aligned}
	      & \collected_a = \False, \\[-1ex]
	      & \forall x \in T(a),\Call{Find}{x} = U
	    \end{aligned} \!\biggr\}$ \label{scc3:f_assign2}
      \sForAll{$a \in F$} $\collected_a \gets \True$\label{scc3:end_atom2}\label{scc3:collected2}
      \sIf{$F \neq \emptyset$} go to Line~\lineref{scc3:begin_edge_loop}\label{scc3:goto}\label{scc3:begin_not_single}
    \EndIf
    \Repeat \label{scc3:begin_non_max_scc_loop}\label{scc3:begin_atom3}
      \State pop $V$ from $S$, add $V$ to $\incomp$
    \Until{$\tindex[V] = \tindex[U]$} \label{scc3:end_non_max_scc_loop}
  \EndIf
\EndFunction\label{scc3:end_atom3}
\end{algorithmic}
\end{minipage}
\end{scriptsize}
\begin{tikzpicture}[remember picture,overlay]
\path (t) ++ (0ex,2ex) coordinate (t);
\path (b) ++ (0ex,-0.5ex) coordinate (b);
\path (l) ++ (-0.5ex,0ex) coordinate (l);
\path (r) ++ (0.5ex,0ex) coordinate (r);
\path let \p1 = (l), \p2 = (t) in coordinate (lt) at (\x1,\y2);
\path let \p1 = (l), \p2 = (b) in coordinate (lb) at (\x1,\y2);
\path let \p1 = (r), \p2 = (t) in coordinate (rt) at (\x1,\y2);
\path let \p1 = (r), \p2 = (b) in coordinate (rb) at (\x1,\y2);
\filldraw[draw=none, style=very nearly transparent, fill=gray!70!black] (lt) -- (rt) -- (rb) -- (lb) -- cycle;

\path (t2) ++ (0ex,-0.5ex) coordinate (t2);
\path (b2) ++ (0ex,-0.5ex) coordinate (b2);
\path (l2) ++ (-0.5ex,0ex) coordinate (l2);
\path (r2) ++ (0.5ex,0ex) coordinate (r2);
\path let \p1 = (l2), \p2 = (t2) in coordinate (lt2) at (\x1,\y2);
\path let \p1 = (l2), \p2 = (b2) in coordinate (lb2) at (\x1,\y2);
\path let \p1 = (r2), \p2 = (t2) in coordinate (rt2) at (\x1,\y2);
\path let \p1 = (r2), \p2 = (b2) in coordinate (rb2) at (\x1,\y2);
\filldraw[draw=none, style=very nearly transparent, fill=gray!70!black] (lt2) -- (rt2) -- (rb2) -- (lb2) -- cycle;

\path (t3) ++ (0ex,-0.5ex) coordinate (t3);
\path (b3) ++ (0ex,-0.5ex) coordinate (b3);
\path (l3) ++ (-0.5ex,0ex) coordinate (l3);
\path (r3) ++ (0.5ex,0ex) coordinate (r3);
\path let \p1 = (l3), \p3 = (t3) in coordinate (lt3) at (\x1,\y3);
\path let \p1 = (l3), \p3 = (b3) in coordinate (lb3) at (\x1,\y3);
\path let \p1 = (r3), \p3 = (t3) in coordinate (rt3) at (\x1,\y3);
\path let \p1 = (r3), \p3 = (b3) in coordinate (rb3) at (\x1,\y3);
\filldraw[draw=none, style=very nearly transparent, fill=gray!70!black] (lt3) -- (rt3) -- (rb3) -- (lb3) -- cycle;

\path (t4) ++ (0ex,-0.5ex) coordinate (t4);
\path (b4) ++ (0ex,-0.5ex) coordinate (b4);
\path (l4) ++ (-0.5ex,0ex) coordinate (l4);
\path (r4) ++ (0.5ex,0ex) coordinate (r4);
\path let \p1 = (l4), \p3 = (t4) in coordinate (lt4) at (\x1,\y3);
\path let \p1 = (l4), \p3 = (b4) in coordinate (lb4) at (\x1,\y3);
\path let \p1 = (r4), \p3 = (t4) in coordinate (rt4) at (\x1,\y3);
\path let \p1 = (r4), \p3 = (b4) in coordinate (rb4) at (\x1,\y3);
\filldraw[draw=none, style=very nearly transparent, fill=gray!70!black] (lt4) -- (rt4) -- (rb4) -- (lb4) -- cycle;
\end{tikzpicture}
\caption{Second intermediary form of our algorithm computing the terminal \scc{}s}\label{fig:maxscc3}
\end{figure}

The second intermediary version of our algorithm, \Call{TerminalScc3}{}, is based on the first one, but it performs the same computations on the auxiliary data $r_a$ and $c_a$ as in \Call{TerminalScc}{}. However, the latter are never used, because at Line~\lineref{scc3:f_assign2}, $F$ is re-assigned to the value provided in \Call{TerminalScc2}{}. It follows that for now, the parts in gray can be ignored. The following lemma states that \Call{TerminalScc2}{} and \Call{TerminalScc3}{} are equivalent:
\begin{proposition}\label{prop:maxscc3}
Let $\HH$ be a directed hypergraph. After the execution of the algorithm $\Call{TerminalScc3}{\HH}$, the sets $\{ v \in \VV \mid \Call{Find}{v} = U \text{ and } \ismax[U] = \True\}$ precisely correspond to the terminal \scc{}s of $\HH$.
\end{proposition}

\begin{proof}
When \Call{Visit3}{$u$} is executed, the local stack $F$ is not directly assigned to the set $\{ a \in A \mid T(a) = \{ u \} \}$ (see Line~\lineref{scc2:f_assign} in Figure~\ref{fig:maxscc2}), but built by several iterations on the set $A_u$ (Line~\lineref{scc3:f_push}). Since $u \in T(a)$ and $\card{T(a)} = 1$ holds if, and only if, $T(a)$ is reduced to $\{ u \}$, \Call{Visit3}{$u$} initially fills $F$ with the same hyperarcs as \Call{Visit2}{$u$}.

Besides, the condition $\single = \False$ in \Call{Visit2}{} (Line~\lineref{scc2:begin_not_single}) is replaced by $F \neq \emptyset$ (Line~\lineref{scc3:begin_not_single}). We claim that the condition $F \neq \emptyset$ can be safely used in \Call{Visit2}{} as well. Indeed, in \Call{Visit2}{}, $F \neq \emptyset$ implies $\single = \False$. Conversely, suppose that in \Call{Visit2}{}, $\single = \False$ and $F = \emptyset$, so that the algorithm goes back to Line~\lineref{scc2:goto} after having $\single$ to $\True$. The loop from Lines~\lineref{scc2:begin_edge_loop} to~\lineref{scc2:end_edge_loop} is not executed since $F = \emptyset$, and it directly leads to a new execution of Lines~\lineref{scc2:is_root} to~\lineref{scc2:begin_not_single} with $\single = \True$. Therefore, going back to Line~\lineref{scc2:goto} was useless.

Finally, during the vertex merging step in \Call{Visit3}{}, $n$ keeps its value, which is greater than or equal to $i+1$, but is not necessarily equal to $i+1$ like in \Call{Visit2}{} (just after Line~\lineref{scc2:nredefined}). This is safe because the whole algorithm only need that $n$ take increasing values, and not necessarily consecutive ones.

We conclude by applying Proposition~\ref{prop:maxscc2}.
\end{proof}

We make similar assumptions on the atomicity of the sequences of instructions. Note that Invariant~\ref{inv:tindex},~\ref{inv:comporstack}, and~\ref{inv:incomp} still holds in \Call{Visit3}{}.

\begin{invariant}\label{inv:re}
Let $a \in A$ such that $\card{T(a)} > 1$. If for all $x \in T(a)$, $\tindex[x]$ is defined, then the root $r_a$ is defined.
\end{invariant}

\begin{proof}
For all $x \in T(a)$, $\Call{Visit3}{x}$ has been called. The root $r_a$ has necessarily been defined at the first of these calls (remember that the block from Lines~\lineref{scc3:begin_atom1} to~\lineref{scc3:end_atom1} is supposed to be executed atomically).
\end{proof}

\begin{invariant}\label{inv:incomp2}
Consider a state $\cur$ of the algorithm in which $U \in \incomp$. Then any vertex reachable from $U$ in $\graph(\HH_\cur)$ is also in $\incomp$. 
\end{invariant}

\begin{proof}
The invariant clearly holds when $U$ is placed in $\incomp$. Using the atomicity assumptions, the call to $\Call{Visit3}{u}$ is necessarily terminated. Let $\old$ be the state of the algorithm at that point, and $\HH_\old$ and $\incomp_\old$ the corresponding hypergraph and set of terminated vertices at that state respectively. Since $\Call{Visit3}{u}$ has performed a depth-first search from the vertex $U$ in $\graph(\HH_\old)$, all the vertices reachable from $U$ in $\HH_\old$ stand in $\incomp_\old$.

We claim that the invariant is then preserved by the following vertex merging steps. The graph arcs which may be added by the latter leave vertices in $S$, and consequently not from elements in $\incomp$ (by Invariant~\ref{inv:comporstack}). It follows that the set of reachable vertices from elements of $\incomp_\old$ is not changed by future vertex merging steps. As a result, \emph{all the vertices reachable from $U$ in $\graph(\HH_\cur)$ are elements of $\incomp_\old$}. Since by Invariant~\ref{inv:incomp2}, $\incomp_\old \subseteq \incomp$, this proves the whole invariant in the state $\cur$.
\end{proof}

\begin{invariant}\label{inv:call_to_visit3}
In the digraph $\graph(\HH_\cur)$, at the call to $\Call{Visit3}{u}$, $u$ is reachable from a vertex $W$ such that $\tindex[W]$ is defined if, and only if, $W$ belongs to the stack $S$.
\end{invariant}

\begin{proof}
The ``if'' part can be shown by induction. When the function $\Call{Visit3}{u}$ is called from Line~\lineref{scc3:init_call}, the stack $S$ is empty, so that this is obvious. Otherwise, it is called from Line~\lineref{scc3:rec_call} during the execution of $\Call{Visit3}{x}$. Then $X  = \Call{Find}{x}$ is reachable from any vertex in the stack, since $x$ was itself reachable from any vertex in the stack at the call to $\Call{Find}{X}$ (inductive hypothesis) and that this reachability property is preserved by potential vertex merging steps (Proposition~\ref{prop:collapse}). As $u$ is obviously reachable from $X$, this shows the statement.

Conversely, suppose that $\tindex[W]$ is defined, and $W$ is not in the stack. According to Invariant~\ref{inv:comporstack}, $W$ is necessarily an element of $\incomp$. Hence $u$ also belongs to $\incomp$ by Invariant~\ref{inv:incomp2}, which is a contradiction since this cannot hold at the call to $\Call{Visit}{u}$.
\end{proof}

\begin{invariant}\label{inv:ce}
Let $a \in A$ such that $\card{T(a)} > 1$. Consider a state $\cur$ of the algorithm \Call{TerminalScc3}{} in which $r_a$ is defined. 

Then $c_a$ is equal to the number of elements $x \in T(a)$ such that $\tindex[x]$ is defined and $\Find{x}$ is reachable from $\Find{r_a}$ in $\graph(\HH_\cur)$.
\end{invariant}

\begin{proof}
Since at Line~\lineref{scc3:counter_increment}, $c_a$ is incremented only if $R_a = \Find{r_a}$ belongs to $S$, we already know using Invariant~\ref{inv:call_to_visit3} that $c_a$ is equal to the number of elements $x \in T(a)$ such that, at the call to $\Call{Visit3}{x}$, $x$ was reachable from $\Find{r_a}$.

Now, let $x \in \VV$, and consider a state $\cur$ of the algorithm in which $r_a$ and $\tindex[x]$ are both defined, and $\Find{r_a}$ appears in the stack $S$. Since $\tindex[x]$ is defined, \Call{Visit3}{} has been called on $x$, and let $\old$ be the state of the algorithm at that point. Let us denote by $\HH_\old$ and $\HH_\cur$ the current hypergraphs at the states $\old$ and $\cur$ respectively. Like previously, we may add a hypergraph as last argument of the function $\Call{Find}{}$ to distinguish its execution in the states $\old$ and $\cur$. We claim that $\Find{r_a,\HH_\cur} \reach_{\graph(\HH_\cur)} \Find{x,\HH_\cur}$ if, and only if, $\Find{r_a,\HH_\old} \reach_{\graph(\HH_\old)} x$. The ``if'' part is due to the fact that reachability in $\graph(\HH_\old)$ is not altered by the vertex merging steps (Proposition~\ref{prop:collapse}). Conversely, if $x$ is not reachable from $\Find{r_a,\HH_\old}$ in $\HH_\old$, then $\Find{r_a,\HH_\old}$ is not in the call stack $S_\old$ (Invariant~\ref{inv:call_to_visit3}), so that it is an element of $\incomp_\old$. But $\incomp_\old \subseteq \incomp_\cur$, which contradicts our assumption since by Invariant~\ref{inv:comporstack}, an element cannot be stored in $\incomp_\cur$ and $S_\cur$ at the same time. It follows that if $r_a$ is defined and $\Find{r_a}$ appears in the stack $S$, $c_a$ is equal to the number of elements $x \in T(a)$ such that $\tindex[x]$ is defined and $\Find{r_a} \reach_{\graph(\HH_\cur)} \Find{x}$. 

Let $\cur$ be the state of the algorithm when $\Find{r_a}$ is moved from $S$ to $\incomp$. The invariant still holds. Besides, in the future states $\new$, $c_a$ is not incremented because $\Find{r_a, \HH_\cur} \in \incomp_\cur \subseteq \incomp_\new$ (Invariant~\ref{inv:incomp}), so that $\Find{r_a,\HH_\new} = \Find{r_a, \HH_\cur}$, and the latter cannot appear in the stack $S_\new$ (Invariant~\ref{inv:comporstack}). Furthermore, any vertex reachable from $R_a = \Find{r_a,\HH_\new}$ in $\graph(\HH_\new)$ belongs to $\incomp_\new$ (Invariant~\ref{inv:incomp2}). It even belongs to $\incomp_\cur$, as shown in the second part of the proof of Invariant~\ref{inv:incomp2} (emphasized sentence). It follows that the number of reachable vertices from $\Find{r_a}$ has not changed between states $\cur$ and $\new$. Therefore, the invariant on $c_a$ will be preserved, which completes the proof.
\end{proof}

\begin{proposition}\label{prop:visit3}
In \Call{Visit3}{}, the assignment at Line~\lineref{scc3:f_assign2} does not change the value of $F$.
\end{proposition}

\begin{proof}
It can be shown by strong induction on the number $p$ of times that this line has been executed. 
Suppose that we are currently at Line~\lineref{scc3:begin_node_merging}, and let $X_1,\dots,X_q$ be the elements of the stack located above the root $U = X_1$ of the terminal \scc{} of $\graph(\HH_\cur)$. Any arc $a$ which will transferred to $F$ from Line~\lineref{scc3:begin_node_merging} to Line~\lineref{scc3:end_node_merging} satisfies $c_a = \card{T(a)} > 1$ and $\Find{r_a} = X_i$ for some $1 \leq i \leq q$ (since at~\lineref{scc3:begin_node_merging}, $F$ is initially empty). Invariant~\ref{inv:ce} implies that for all elements $x \in T(a)$, $\Find{x}$ is reachable from $X_i$ in $\graph(\HH_\cur)$, so that by terminality of the \scc{} $C = \{X_1,\dots,X_q\}$, $\Find{x}$ belongs to $C$, \ie\ there exists $j$ such that $\Find{x} = X_j$. It follows that at Line~\lineref{scc3:end_node_merging}, $\Find{x} = U$ for all $x \in T(a)$. Then, we claim that $\collected_a = \False$ at Line~\lineref{scc3:end_node_merging}. Indeed, $a' \in A$ satisfies $\collected_{a'} = \True$ if, and only if: 
\begin{itemize}[\textbullet]
\item either it has been copied to $F$ at Line~\lineref{scc3:f_push}, in which case $\card{T(a')} = 1$,
\item or it has been copied to $F$ at the $r$-th execution of Line~\lineref{scc3:f_assign2}, with $r < p$. By induction hypothesis, this means that $a'$ has been pushed on a stack $F_X$ and then popped from it strictly before the $r$-th execution of Line~\lineref{scc3:f_assign2}.
\end{itemize}
Observe that a given hyperarc can be popped from a stack $F_x$ at most once during the whole execution of \Call{TerminalScc3}{}. Here, $a$ has been popped from $F_{X_i}$ after the $p$-th execution of Line~\lineref{scc3:f_assign2}, and $\card{T(a)} > 1$. It follows that $\collected_a = \False$. 

Conversely, suppose for that, at Line~\lineref{scc3:f_assign2}, $\collected_a = \False$, and all the $x \in T(a)$ satisfies $\Find{x} = U$. Clearly, $\card{T(a)} > 1$ (otherwise, $a$ would have been placed into $F$ at Line~\lineref{scc3:f_push} and $\collected_a$ would be equal to $\True$). Few steps before, at Line~\lineref{scc3:begin_node_merging}, $\Find{x}$ is equal to one of $X_j$, $1 \leq j \leq q$. Since $\tindex[X_j]$ is defined ($X_j$ is an element of the stack $S$), by Invariant~\ref{inv:tindex}, $\tindex[x]$ is also defined for all $x \in T(a)$, hence, the root $r_a$ is defined by Invariant~\ref{inv:re}. Besides, $\Find{r_a}$ is equal to one of the $X_j$, say $X_k$ (since $r_a \in T(a)$). As all the $\Find{x}$ are reachable from $\Find{r_a}$ in $\graph(\HH_\cur)$, then $c_a = \card{T(a)}$ using Invariant~\ref{inv:ce}. It follows that $a$ has been pushed on the stack $F_{R_a}$, where $R_a = \Find{r_a,\HH_\old}$ in an previous state $\old$ of the algorithm. As $\collected_a = \False$, $a$ has not been popped from $F_{R_a}$, and consequently, the vertex $R_a$ of $\HH_\old$ has not involved in a vertx merging step. Therefore, $R_a$ is still equal to $\Find{r_a,\HH_\cur} = X_k$. It follows that at Line~\lineref{scc3:begin_node_merging}, $a$ is stored in $F_{X_k}$, and thus it is copied to $F$ between Lines~\lineref{scc3:begin_node_merging} and~\lineref{scc3:end_node_merging}. This completes the proof.
\end{proof}

We now can prove the correctness of \Call{TerminalScc}{}.
\begin{proof}[Theorem~\ref{th:correctness}]
By Proposition~\ref{prop:visit3}, Line~\lineref{scc3:f_assign2} can be safely removed in \Call{Visit3}{}. It follows that the booleans $\collected_a$ are now useless, so that Line~\lineref{scc3:collected_init}, the loop from Lines~\lineref{scc3:collected1_begin} to~\lineref{scc3:collected1_end}, and Line~\lineref{scc3:collected2} can be also removed. After that, we precisely obtain the algorithm \Call{TerminalScc}{}. Proposition~\ref{prop:maxscc3} completes the proof.
\end{proof}

\end{document}